\documentclass[sigconf,authorversion,nonacm]{aamas}

\usepackage{balance} 

\usepackage[capitalize]{cleveref}
\usepackage{enumerate}
\usepackage{mathtools}
\usepackage{thm-restate}
\usepackage{xargs}

\setcopyright{ifaamas}
\acmConference[AAMAS '26]{Proc.\@ of the 25th International Conference
on Autonomous Agents and Multiagent Systems (AAMAS 2026)}{May 25 -- 29, 2026}
{Paphos, Cyprus}{C.~Amato, L.~Dennis, V.~Mascardi, J.~Thangarajah (eds.)}
\copyrightyear{2026}
\acmYear{2026}
\acmDOI{}
\acmPrice{}
\acmISBN{}

\acmSubmissionID{47}

\title[Comparing the Fairness of Recursively Balanced Picking Sequences]{Comparing the Fairness of \\Recursively Balanced Picking Sequences}

\author{Karen Frilya Celine}
\affiliation{
  \institution{National University of Singapore}
  \country{Singapore}
}

\author{Warut Suksompong}
\affiliation{
  \institution{National University of Singapore}
  \country{Singapore}
}

\author{Sheung Man Yuen}
\affiliation{
  \institution{National University of Singapore}
  \country{Singapore}
}

\begin{abstract}
Picking sequences are well-established methods for allocating indivisible goods. 
Among the various picking sequences, \emph{recursively balanced picking sequences}---whereby each agent picks one good in every round---are notable for guaranteeing allocations that satisfy envy-freeness up to one good.
In this paper, we compare the fairness of different recursively balanced picking sequences using two key measures.
Firstly, we demonstrate that all such sequences have the same price in terms of egalitarian welfare relative to other picking sequences.
Secondly, we characterize the approximate maximin share (MMS) guarantees of these sequences. 
In particular, we show that compensating the agent who picks last in the first round by letting her pick first in every subsequent round yields the best MMS guarantee.
\end{abstract}

\keywords{Picking sequence, Egalitarian welfare, Maximin share, Fair division}

\newcommand{\BibTeX}{\rm B\kern-.05em{\sc i\kern-.025em b}\kern-.08em\TeX}

\allowdisplaybreaks

\newcommand{\EW}{\textup{EW}}
\newcommand{\MMS}{\textup{MMS}}
\newcommand{\roundrobin}{\pi_{\textup{RR}}}
\newcommand{\pickseq}{\Pi_{n, m}}
\newcommandx{\recbalseq}[2][1=n,2=m]{\mathcal{R}_{#1, #2}}

\renewcommand{\emptyset}{\varnothing}

\makeatletter
\def\old@comma{,}
 \catcode`\,=13
 \def,{%
   \ifmmode%
     \old@comma\discretionary{}{}{}%
   \else%
     \old@comma%
   \fi%
 }
 \makeatother

\begin{document}


\pagestyle{fancy}
\fancyhead{}
\pagestyle{plain} 


\maketitle 


\section{Introduction}
\label{sec:intro}

As humans living together and sharing limited resources, we are bound to face the problem of how to fairly divide the resources, also known as \emph{fair division} \citep{Moulin03}. 
An important application of fair division is the allocation of \emph{indivisible goods}, such as allocating course slots to students, assigning sports players to teams, or distributing research equipment among labs \citep{DemkoHi88,BramsEdFi03,AmanatidisAzBi23}.

While numerous methods have been proposed for fairly allocating indivisible goods, one of the most fundamental class of methods is that of \emph{picking sequences} \citep{KohlerCh71,BramsTa00,BouveretLa14}.
Picking sequences allow agents to select goods according to a prespecified agent order---at each turn, the designated agent chooses her favorite good from the remaining goods.
Not only are picking sequences intuitive and easy to implement, but they also help preserve the privacy of the participating agents, as the agents only need to reveal their picks rather than their entire valuations for the goods.
A picking sequence is called \emph{recursively balanced} if at every point in the sequence, the difference between the number of turns taken by each pair of agents is at most one.
Note that a recursively balanced picking sequence can be divided into \emph{rounds}, where each agent picks exactly once in every round (except possibly the last round, in which some agents may not receive a pick).
Recursively balanced picking sequences are notable because they always produce allocations satisfying the fairness notion of \emph{envy-freeness up to one good (EF1)}, which means that if an agent envies another agent, then the envy can be eliminated by removing a good from the latter agent's bundle \citep[p.~5]{AmanatidisAzBi23}.\footnote{This guarantee relies on the assumption that agents have additive valuations over the goods. 
This assumption is common in the fair division literature, and we will make it throughout our paper.}

Among recursively balanced picking sequences, the most widely studied one is \emph{round-robin}, which lets the agents pick in the order $(1, 2, \ldots, n \mid 1, 2, \ldots, n \mid 1, 2, \ldots)$, where $n$ denotes the number of agents.
However, round-robin is far from the only recursively balanced picking sequence.
Another natural recursively balanced sequence is \emph{balanced alternation}, which reverses the ordering in every alternate round: $(1, 2, \ldots, n \mid n, n-1, \ldots, 1 \mid 1, 2, \ldots)$.\footnote{The term \emph{balanced alternation} has also been used to refer to other picking sequences \citep{BramsTa00,BramsIs21}.}
Intuitively, balanced alternation appears fairer than round-robin, as it gives the agents more equal opportunities to choose their preferred goods.
The applicability of balanced alternation is demonstrated by the fact that it is used to allocate courses to students at Harvard Business School \citep{BudishCa12}.
Nevertheless, given that all recursively balanced picking sequences ensure EF1, it is unclear which sequence should be considered the ``fairest''.
Are there recursively balanced picking sequences that fare better with respect to certain fairness criteria than these well-known sequences?

In this paper, we compare the fairness of recursively balanced picking sequences using two established measures: egalitarian welfare and maximin share (MMS).

\subsection{Our Results}
\label{subsec:results}

Assume that there are $n$ agents with additive valuations over $m$ indivisible goods.
If $m < n$, all recursively balanced picking sequences are effectively equivalent, so we assume that $m \ge n$.
Without loss of generality, we consider picking sequences starting with the prefix $(1, 2, \ldots, n)$. 
Our model is described formally in \cref{sec:prelim}.

In \Cref{sec:ew}, we compare recursively balanced picking sequences based on their worst-case egalitarian welfare relative to other picking sequences.
Specifically, we define the \emph{egalitarian price} of a picking sequence as the supremum, taken across all possible instances, of the ratio between the ``optimal egalitarian welfare'' and the egalitarian welfare of the given picking sequence.
We show that this ratio is infinite for any recursively balanced picking sequence if the optimal egalitarian welfare is taken among \emph{all allocations} or \emph{all picking sequences} within the instance, which renders the comparison impractical.
Hence, in \cref{subsec:ew_allpickseq}, we take the optimal egalitarian welfare among allocations obtained via picking sequences with the same prefix $(1, 2, \ldots, n)$. 
In \cref{subsec:ew_recbalseq}, we further restrict the optimal welfare to be across allocations produced by \emph{recursively balanced} picking sequences starting with $(1, 2, \ldots, n)$.
Perhaps surprisingly, for both variants, we find that all recursively balanced picking sequences have the same egalitarian price: $\min\{m-n+1, n\}$ and $\min\{\lceil m/n\rceil, \lfloor \log_2n\rfloor + 1\}$ respectively.
For the latter variant, our proof involves analyzing paths along a directed graph constructed from the agents' picks.

Next, in \cref{sec:mms}, we compare recursively balanced picking sequences using their approximate MMS guarantees.
To state these guarantees, we define conditions for a picking sequence to be \emph{regular}.
Most picking sequences---including all sequences with $m \geq 2n$---are regular, and we characterize their MMS guarantees in \cref{thm:mms_regular}. 
For such sequences, the MMS guarantee depends only on the picks of agent~$n$ (who picks last in the first round).
However, there are a small number of \emph{irregular} picking sequences, for which the agent with the worst MMS guarantee is agent $n-1$ instead of agent~$n$.
The MMS guarantees of these sequences are characterized in \cref{thm:mms_irregular}.
Our characterizations allow us to determine
the best picking sequences with respect to MMS in \cref{thm:mms_best}; these include the sequences that compensate agent~$n$ by letting her pick first in every round after the first.
On the other end of the spectrum, we also identify the picking sequences with the worst MMS guarantee in \cref{thm:mms_worst}. 
While these include the round-robin sequence as one may expect, more interestingly, they also include the balanced alternation sequence whenever $m \geq 3n-1$.

\subsection{Further Related Work}
\label{subsec:related}

As mentioned earlier, picking sequences have long been studied in fair division \citep{KohlerCh71,BramsTa00,BouveretLa11}. 
Since it is sometimes beneficial for an agent to avoid picking her favorite good if she is aware of other agents' preferences, several authors have investigated picking sequences from a strategic perspective \citep{KalinowskiNaWa13b,BouveretLa14,TominagaToYo16,AzizBoLa17}.
Within the class of picking sequences, round-robin has received particular attention due to its simplicity and EF1 fairness guarantee \citep{AmanatidisBiFu24,LiMaSc25}.
\citet{BouveretGiLa25} focused on \emph{constrained serial dictatorships}, where the turns of each agent occur consecutively, e.g., $(1, 2, 2, 3, 3, 3, 3)$.
While such picking sequences are strategyproof, they are generally far from guaranteeing EF1.
\citet{ChakrabortyScSu21} showed that picking sequences provide meaningful fairness guarantees when agents have different entitlements, while \citet{GourvesLeWi21} defined fairness criteria based on the picking sequences themselves.
\citet{AzizFrSh24} attained ``best-of-both-worlds'' fairness via a lottery over picking sequences.

Given the variety of picking sequences, a natural direction is to compare them with respect to particular criteria. 
\citet{BouveretLa11} studied the \emph{expected} utilitarian and egalitarian welfare of picking sequences, and computed the optimal picking sequences for small numbers of agents and goods (see their Table~1).
Subsequently, \citet{KalinowskiNaWa13a} proved that round-robin yields the optimal expected utilitarian welfare when there are two agents, under certain distributions of the agents' utilities (see their Theorem~1).
Unlike these average-case analyses, our results are worst-case and do not rely on any distributional assumptions.
\citet{AzizWaXi15} examined the complexity of possible and necessary allocation problems for different classes of picking sequences. 

Our egalitarian price notion is inspired by the well-established \emph{price of fairness} concept, which captures the ratio between the optimal welfare overall and the optimal welfare subject to a given fairness requirement such as EF1 \citep{CaragiannisKaKa12,BeiLuMa21,CelineDzKo23,LiLiLu24}.
While both egalitarian and utilitarian welfare have been considered in this line of work, we focus on egalitarian welfare in our paper since it is widely regarded as a fairness measure (whereas utilitarian welfare is typically viewed as a measure of efficiency).
\citet[Sec.~6.2]{BaumeisterBoLa17} also studied the price of picking sequences but restricted utilities to follow a specific scoring vector.
Finally, approximate MMS has received significant attention in recent years \citep{KurokawaPrWa18,GhodsiHaSe21,AkramiGa24}.

\section{Preliminaries}
\label{sec:prelim}

Let $N = \{1, \ldots, n\}$ be a set of $n \geq 2$ agents, and $M$ be a set of $m$ goods; we typically denote the goods by $g_1, \ldots, g_m$.
A~\emph{bundle} refers to a (possibly empty) set of goods in~$M$.
Each agent $i\in N$ has a utility function $u_i$ such that $u_i(S)$ is agent~$i$'s utility for the bundle $S \subseteq M$; we write $u_i(g)$ instead of $u_i(\{g\})$ for a single good $g \in M$.
Each utility function is \emph{additive}, i.e., $u_i(S) = \sum_{g \in S} u_i(g)$ for any $S \subseteq M$.
The utilities are \emph{identical} if $u_i = u_j$ for all $i, j \in N$. 
Furthermore, each agent $i$ has a \emph{picking preference order} $\succ_i$, which is a total order on the set of goods~$M$: for any $g, g' \in M$, $g \succ_i g'$ means that agent $i$ would choose good~$g$ before good $g'$, provided both goods are available. 
We assume that $u_i(g) > u_i(g')$ implies $g \succ_i g'$ for any $g, g' \in M$.
We require $\succ_i$ to be a total order to facilitate tie-breaking.\footnote{Alternatively, one could assume that every agent always breaks ties by picking a good with a lower index before one with a higher index if the agent values both goods equally.
All specific instances we use in our proofs satisfy this property.}
An \emph{instance} consists of $N$, $M$, $(u_i)_{i \in N}$, and $(\succ_i)_{i\in N}$.
Denote by $\mathcal{I}_{n, m}$ the set of all instances with $n$ agents and $m$ goods.

Given an instance $\mathcal{I} \in \mathcal{I}_{n, m}$, let $\mathcal{P}$ be the set of all partitions of $M$ into $n$ bundles. The \emph{maximin share (MMS) of agent $i$} is defined by
\begin{align*}
    \MMS_i \coloneqq \max_{\{P_1, \ldots, P_n\} \in \mathcal{P}} \min_{j\in\{1,\dots,n\}} u_i(P_j).
\end{align*}
An allocation $\mathcal{A} = (A_1, \ldots, A_n)$ is an ordered partition of~$M$ into $n$ bundles $A_1, \ldots, A_n$ such that $A_i$ is allocated to agent $i \in N$.
An allocation~$\mathcal{A}$ is called \emph{envy-free up to one good (EF1)} if for any $i, j \in N$ with $A_j \neq \emptyset$, there exists $g \in A_j$ such that $u_i(A_i) \geq u_i(A_j \setminus \{g\})$.
The \emph{egalitarian welfare} of allocation $\mathcal{A}$ for instance $\mathcal{I}$ is defined to be $\EW(\mathcal{A}, \mathcal{I}) \coloneqq \min_{i \in N} u_i(A_i)$.

A \emph{picking sequence} is a sequence $\pi = (a_1, \ldots, a_m)$, where $a_j \in N$ for each $j \in \{1, \ldots, m\}$.
For each agent $i \in N$, \emph{the picking sequence of agent $i$ in $\pi$} is defined as $\pi_i = (t_1, \ldots, t_R)$, where $a_t = i$ if and only if $t \in \{t_1, \ldots, t_R\}$, and $t_1 < \dots < t_R$.
For each $r \in \{1, \ldots, R\}$, we call $t_r$ \emph{the index of agent $i$'s $r$-th pick in $\pi$}.

A picking sequence $\pi$ is \emph{recursively balanced} if for every prefix of~$\pi$ and every pair of agents, the difference in the number of times that the two agents appear in the prefix is at most~$1$.
Note that for each $j \in \{1, \ldots, \lfloor m/n \rfloor\}$, the prefix of length $jn$ of a recursively balanced picking sequence $\pi$ contains each agent exactly $j$ times. Hence, the subsequence $(a_{(j-1)n+1}, \ldots, a_{jn})$ contains each agent exactly once; we call this subsequence \emph{the $j$-th round} of $\pi$.
When $m$ is not divisible by $n$, \emph{the $\lceil m/n \rceil$-th round} of $\pi$ is the subsequence $(a_{\lfloor m/n \rfloor n+1}, \ldots, a_m)$, which contains each agent at most once.
For clarity, we may use a vertical bar ($\mid$) instead of a comma ($,$) to separate the sequences of different rounds. 
That is, given a picking sequence $\pi$, we may write $\pi = (a_1, \ldots, a_n \mid a_{n+1}, \ldots, a_{2n} \mid \cdots \mid a_{\lfloor m/n \rfloor n + 1}, \ldots, a_m)$.
We sometimes use the term \emph{round} to denote the same picks for picking sequences that are not recursively balanced---for such sequences, it is not necessary that every agent picks at most once in each round.
The \emph{round-robin sequence} is the recursively balanced picking sequence
\[
    \roundrobin \coloneqq (1, 2, \ldots, n \mid \cdots \mid 1, 2, \ldots, n \mid 1, 2, \ldots, a_m),
\]
where $a_m = m - (\lceil m/n \rceil - 1)n$.

We denote by $\pickseq$ the set of all picking sequences with $n$ agents and $m$ goods prefixed by $(1, 2, \ldots, n)$.
Moreover, we denote by $\recbalseq$ the set of all \emph{recursively balanced} picking sequences with $n$ agents and $m$ goods prefixed by $(1, 2, \ldots, n)$.
Note that $\pickseq$ (resp.~$\recbalseq$) does not contain all picking sequences (resp.~recursively balanced picking sequences). 
However, for any picking sequence (resp.~recursively balanced picking sequence) $\pi$ with all agents in the first $n$ picks, there exists a picking sequence $\pi'$ in $\pickseq$ (resp.~$\recbalseq$) that is equivalent to $\pi$ up to some relabelling of the agents.

Given an instance and a picking sequence $\pi = (a_1, \ldots, a_m)$, the allocation $\mathcal{A}^\pi = (A_1^\pi, \ldots, A_n^\pi)$ obtained with $\pi$ is given as follows: each agent begins with an empty bundle, and at each step $j \in \{1, \ldots, m\}$, agent $a_j$ selects the $\succ_{a_j}$-maximum good available among the unallocated goods and adds this good to her bundle.
The \emph{round-robin allocation} is the allocation obtained with $\roundrobin$.
Given a picking sequence $\pi$ and an instance $\mathcal{I}$, we write $\EW(\pi, \mathcal{I}) \coloneqq \EW(\mathcal{A}^\pi, \mathcal{I})$ to mean the egalitarian welfare of the allocation obtained by $\pi$.

We make the following observation.

\begin{restatable}{proposition}{propEfone}
\label{prop:EF1}
A picking sequence always produces an EF1 allocation if and only if it is recursively balanced.
\end{restatable}

\begin{proof}
The backward direction has been shown by \citet[p.~5]{AmanatidisAzBi23}.
For the forward direction, consider a picking sequence~$\pi$ which is not recursively balanced.
This means that there exists a prefix $\rho$ of $\pi$ and a pair of agents $i, j$ such that agent $j$ picks at least two more times than agent $i$ in $\rho$.
Let $\ell$ be the length of the prefix $\rho$.
Consider an instance where agents $i$ and $j$ value each of $g_1,\dots,g_\ell$ at~$1$ and the remaining goods at $0$, all other agents value each good at~$1$, and agents break ties in favor of lower-index goods.
Then, goods $g_1$ to $g_m$ are chosen in ascending order.
By the definition of~$\rho$, agent $j$ receives at least two more goods of value~$1$ than agent~$i$. 
Hence, agent $i$ envies agent $j$ by more than one good, and the allocation is not EF1.
\end{proof}

\section{Egalitarian Price}
\label{sec:ew}

In this section, we compare recursively balanced picking sequences using the egalitarian welfare.
We first observe that if the utilities are identical, then round-robin is always the worst among such sequences.

\begin{restatable}{proposition}{propEwIdenWorstRr}
\label{prop:ew_iden_worst_rr}
Let $\mathcal{I}$ be an instance with $n$ agents with identical utilities and $m$ goods, and let $\pi \in \recbalseq$.
Then, the egalitarian welfare of the allocation obtained with $\pi$ is at least the corresponding welfare with $\roundrobin$.
\end{restatable}

\begin{proof}
Let $u$ denote the common utility function and assume, without loss of generality, that $u(g_1) \geq u(g_2) \geq \ldots \geq u(g_m)$.
Given any picking sequence $\pi = (a_1, a_2, \ldots, a_m) \in \recbalseq$, each agent $a_j$ must pick a good of value $u(g_j)$. (Note that depending on the agents' tie-breaking mechanisms, the exact good chosen by agent $a_j$ may not be $g_j$.)
In each round $r \in \{1, \ldots, \lfloor m/n \rfloor\}$, each agent $i \in N$ can select some good of value at least $u(g_{rn})$.
Hence, agent $i$ will get a total utility of at least $\sum_{r=1}^{\lfloor m/n \rfloor} u(g_{rn})$.
On the other hand, under the round-robin sequence $\roundrobin$, agent $n$ gets a total utility of exactly $\sum_{r=1}^{\lfloor m/n \rfloor} u(g_{rn})$.
It follows that the egalitarian welfare of the allocation obtained with $\pi$ is at least the egalitarian welfare of the allocation obtained with the round-robin sequence $\roundrobin$.    
\end{proof}

Interestingly, \cref{prop:ew_iden_worst_rr} ceases to hold for non-identical utilities, even with identical orders of preference.

\begin{example}
Consider the following instance.
\begin{center} 
    \begin{tabular}{c|cccc}
        $g$      & $g_1$ & $g_2$ & $g_3$ & $g_4$\\
        \hline
        $u_1(g)$ & $8$ & $7$ & $5$ & $0$ \\
        $u_2(g)$ & $7$ & $6$ & $4$ & $3$
    \end{tabular}
\end{center}
The round-robin allocation is $(\{g_1, g_3\}, \{g_2, g_4\})$, which has an egalitarian welfare of $9$.
However, the allocation obtained with $\pi = (1, 2, 2, 1)$ is $(\{g_1, g_4\}, \{g_2, g_3\})$, which has an egalitarian welfare of only~$8$.
\end{example}

In order to compare different sequences, a natural approach is to consider their egalitarian welfare in the worst case.
However, this approach is not meaningful because for any sequence, this value can be $0$ even when the optimal egalitarian welfare is positive.

\begin{example}
\label{ex:ew}
Given any $n \geq 2$, $m \geq n$, and $\pi \in \recbalseq$, consider an instance with the following utilities:
\begin{itemize}
    \item $u_1(g_1) = 2$, $u_1(g_2) = 1$, and $u_2(g_1) = 3$.
    \item For each agent $i \in N \setminus \{1, 2\}$, we have $u_i(g_i) = 3$.
    \item $u_i(g_j) = 0$ for all other pairs $(i, j)$.
\end{itemize}
Agents break ties in favor of lower-index goods.

Observe that according to $\pi$, the goods will be chosen in increasing order of index, i.e., each agent $i \in N$ picks good~$g_i$ in her first turn, and picks a good with value $0$ in any subsequent turn.
Hence, agent $1$ receives utility $2$, agent~$2$ receives utility $0$, and each agent $i \in N \setminus \{1, 2\}$ receives utility $3$.
Therefore, the egalitarian welfare obtained with~$\pi$ is $0$.
On the other hand, if agents $1$ and $2$ swap $g_1$ and $g_2$, the egalitarian welfare becomes~$1$.
Moreover, this allocation can be obtained with the picking sequence $\pi'$ derived from~$\pi$ by switching the first two picks (between agents $1$ and $2$).
\end{example}

\Cref{ex:ew} implies that the conventional definition of the \emph{egalitarian price} of $\pi \in \recbalseq$, i.e.,
\begin{align*}
    \sup_{\mathcal{I} \in \mathcal{I}_{n, m}} \max_{\mathcal{A}} \frac{\EW(\mathcal{A}, \mathcal{I})}{\EW(\pi, \mathcal{I})},
\end{align*}
where the maximum is taken across all allocations~$\mathcal{A}$ in the instance~$\mathcal{I}$ \citep{AumannDo15,Suksompong19,CelineDzKo23}, is not useful as a metric for comparison.\footnote{For such fractions, we interpret $\frac{0}{0}$ to be equal to $1$.}
Indeed, with this definition, the egalitarian price of any picking sequence $\pi \in \recbalseq$ would be $\infty$ for any $n \geq 2$ and $m \geq n$.
Moreover, the picking sequence~$\pi'$ in \Cref{ex:ew} yields positive egalitarian welfare.\footnote{In fact, \citet{CelineDzKo23} showed that for every instance, there always exists a recursively balanced picking sequence with the same agent ordering in every round, such that the resulting allocation has an egalitarian welfare of at least $1/(2n-1)$ times the optimum.
However, the ordering of agents in each round of such a picking sequence is not fixed for all instances, but may differ depending on the instance.
}
Hence, this also rules out the following definition of the egalitarian price of $\pi \in \recbalseq$ with respect to all other picking sequences:
\begin{align*}
    \sup_{\mathcal{I} \in \mathcal{I}_{n, m}} \max_{\pi'} \frac{\EW(\pi',\mathcal{I})}{\EW(\pi,\mathcal{I})}.
\end{align*}

In a bid to find a meaningful comparison metric, we observe that the picking sequence $\pi'$ in \cref{ex:ew} does not belong to $\pickseq$ or $\recbalseq$, since it does not start with the prefix $(1, 2, \ldots, n)$.
As it turns out, defining the optimal welfare as the maximum welfare over all picking sequences $\pi' \in \pickseq$ or $\pi' \in \recbalseq$ allows us to circumvent the issue that the welfare provided by a given picking sequence~$\pi$ can be~$0$ even when the optimal welfare is positive.
In \cref{subsec:ew_allpickseq}, we consider the egalitarian price of $\pi \in \recbalseq$ with respect to all picking sequences $\pi' \in \pickseq$. In \cref{subsec:ew_recbalseq}, we restrict the search of the optimal egalitarian welfare to recursively balanced picking sequences, and consider the egalitarian price of $\pi \in \recbalseq$ with respect to all recursively balanced picking sequences $\pi' \in \recbalseq$.
Interestingly, we show that for each metric, the egalitarian price of every recursively balanced picking sequence $\pi$ is the same.
This means that all recursively balanced picking sequences are ``equally fair'' with respect to these metrics based on egalitarian welfare.

\subsection{Price Relative to All Picking Sequences}
\label{subsec:ew_allpickseq}

First, we derive a tight bound for the egalitarian price relative to other picking sequences starting with the same $n$ picks, as stated in \cref{thm:ew_allpickseq} below.
Note that the egalitarian price depends only on $n$ and $m$, and not on the specific picking sequence $\pi$.
Therefore, every picking sequence is equally fair with respect to this version of the egalitarian price.

\begin{theorem}
\label{thm:ew_allpickseq}
For any $n \geq 2$, $m \geq n$, and $\pi \in \recbalseq$,
\begin{align*}
    \sup_{\mathcal{I} \in \mathcal{I}_{n, m}} \max_{\pi' \in \pickseq} \frac{\EW(\pi',\mathcal{I})}{\EW(\pi,\mathcal{I})} &= \min\{m-n+1, n\}.
\end{align*}
\end{theorem}

\begin{proof}
We begin by establishing the upper bound.
Let $\pi \in \recbalseq$, $\mathcal{I} \in \mathcal{I}_{n,m}$, and $\pi' \in \pickseq$.

First, we show that when $n \leq m \leq 2n-1$, an upper bound is $\min\{m-n+1, n\} = m-n+1$.
Fix $k \in N$, and let $g$ be the good picked by agent $k$ in the first round of $\pi$.
Note that $g$ is also picked by agent $k$ in the first round of $\pi'$, since both picking sequences have the same prefix for the first round.
Agent $k$ receives a utility of at least $u_k(g)$ in the allocation obtained by $\pi$.
In $\pi'$, agent $k$ receives at most $1+(m-n)$ goods, and the utility of each good is at most $u_k(g)$, so agent $k$ receives a utility of at most $(m-n+1)\cdot u_k(g)$.
Therefore, the ratio of agent $k$'s utility in the allocation obtained by $\pi'$ to the corresponding utility for $\pi$ is at most $m-n+1$.
Since $k \in N$ was arbitrarily chosen, we have
\begin{align*}
    \frac{\EW(\pi',\mathcal{I})}{\EW(\pi,\mathcal{I})} \leq m-n+1,
\end{align*}
proving the upper bound when $n \leq m \leq 2n-1$.

Next, we show that when $m \geq 2n$, an upper bound is $\min\{m-n+1, n\} = n$.
Fix $k \in N$.
Again, the goods picked by the agents in the first round of $\pi$ are identical to that in the first round of $\pi'$.
Without loss of generality, relabel the goods as follows: let $g_i$ be the good picked by agent $i \in N$ in the first round of $\pi$ and~$\pi'$, and let the goods in $M' =\{g_{n+1}, \ldots, g_m\}$ be arranged in descending order of agent $k$'s picking preference, i.e., for all $n+1 \leq j_1, j_2 \leq m$, we have $j_1 < j_2$ if and only if $g_{j_1} \succ_k g_{j_2}$.

Consider the allocation obtained by $\pi$.
For every $r \in \{2, \ldots, \lfloor m/n \rfloor \}$, agent $k$ is able to select some good from $\{g_{n+1}, \ldots, g_{rn}\}$ in round $r$.
This is because $\pi$ is recursively balanced, so when it is agent $k$'s turn in round $r$, at most $rn-1$ goods (including those in $M \setminus M'$) have been taken.
This means that the most preferred good for agent $k$ in round~$r$ is no worse than $g_{rn}$.
Agent $k$'s utility of her bundle from $\pi$ is therefore at least $u_k(g_k) + \sum_{r=2}^{\lfloor m/n \rfloor} u_k(g_{rn})$.

Now, consider the allocation obtained by $\pi'$.
Agent $k$ receives a subset of $\{g_k\} \cup M'$, which has utility at most
\begin{align*}
    &u_k(g_k) + u_k(M') \\
    &= u_k(g_k) + \sum_{j=n+1}^{2n-1} u_k(g_j) + \sum_{r=2}^{\lfloor m/n \rfloor-1} \sum_{j=0}^{n-1} u_k(g_{rn+j}) \\
    &\qquad+ \sum_{j=\lfloor m/n \rfloor n}^m u_k(g_j) \\
    &\leq u_k(g_k) + \sum_{j=n+1}^{2n-1} u_k(g_k) + \sum_{r=2}^{\lfloor m/n \rfloor-1} \sum_{j=0}^{n-1} u_k(g_{rn}) \\
    &\qquad+ \sum_{j=\lfloor m/n \rfloor n}^m u_k(g_{\lfloor m/n \rfloor n}) \\
    &\leq n \cdot u_k(g_k) + \sum_{r=2}^{\lfloor m/n \rfloor-1} n \cdot u_k(g_{rn}) + n \cdot u_k(g_{\lfloor m/n \rfloor n}) \\
    &= n \cdot \left( u_k(g_k) + \sum_{r=2}^{\lfloor m/n \rfloor} u_k(g_{rn}) \right),
\end{align*}
which is at most $n$ multiplied by agent $k$'s utility of her bundle from~$\pi$.
Since $k \in N$ was arbitrarily chosen, we have
\begin{align*}
    \frac{\EW(\pi',\mathcal{I})}{\EW(\pi,\mathcal{I})} \leq n,
\end{align*}
proving the upper bound when $m \geq 2n$.

We next turn to the lower bound.
First, we show that when $n \leq m \leq 2n-1$, a lower bound is $\min\{m-n+1, n\} = m-n+1$.
Let $\pi \in \recbalseq$, and let $k \in N$ be an agent who does not get to pick in the second round---such an agent must exist since $m \leq 2n-1$ and so the second round, if it exists, is incomplete.
Consider an instance~$\mathcal{I} \in \mathcal{I}_{n,m}$ with the following utilities:
\begin{itemize}
    \item $u_k(g_j) = 1/(m-n+1)$ for $j \in \{k\} \cup \{n+1, \ldots, m\}$.
    \item For each agent $i \in N \setminus \{k\}$, we have $u_i(g_i) = 1$.
    \item $u_i(g_j) = 0$ for all other pairs $(i, j)$.
\end{itemize}
Agents break ties in favor of lower-index goods.

Consider the allocation obtained by $\pi$.
In the first round, every agent $i \in N$ selects $g_i$.
In the second round, agent~$k$ does not select any good.
Agent $k$ receives a utility of $1/(m-n+1)$ from $g_k$, and every other agent $i$ receives a utility of $1$ from $g_i$, so the egalitarian welfare is $1/(m-n+1)$.

Now, consider the allocation obtained by $\pi' = (1, 2, \ldots, n, k, k, \ldots, k) \in \pickseq$.
In the first round, every agent $i \in N$ selects $g_i$.
In the second round, agent $k$ selects all of $g_{n+1}, \ldots, g_m$.
Agent $k$ receives a utility of $1$ from picking all $m-n+1$ goods valuable to her, and every other agent~$i$ receives a utility of $1$ from $g_i$, so the egalitarian welfare is~$1$.
This shows that
\begin{align*}
    \frac{\EW(\pi',\mathcal{I})}{\EW(\pi,\mathcal{I})} &= \frac{1}{1/(m-n+1)} = m-n+1,
\end{align*}
proving the lower bound when $n \leq m \leq 2n-1$.

Next, we show that when $m \geq 2n$, a lower bound is $\min\{m-n+1, n\} = n$.
Let $\pi \in \recbalseq$, and let $k = a_{2n}$, i.e., the agent who gets the $(2n)$-th pick in~$\pi$.
Consider an instance $\mathcal{I} \in \mathcal{I}_{n,m}$ with the following utilities:
\begin{itemize}
    \item $u_k(g_j) = 1/n$ for $j \in \{k\} \cup \{n+1, \ldots, 2n-1\}$.
    \item For each agent $i \in N \setminus \{k\}$, we have $u_i(g_i) = 1$.
    \item $u_i(g_j) = 0$ for all other pairs $(i, j)$.
\end{itemize}
Agents break ties in favor of lower-index goods.

Consider the allocation obtained by $\pi$.
In the first round, every agent $i \in N$ selects $g_i$.
In every subsequent round, agent $k$ does not get to select any good valuable to her.
Agent~$k$ receives a utility of $1/n$ from $g_k$, and every other agent $i$ receives a utility of $1$ from $g_i$, so the egalitarian welfare is $1/n$.

Now, consider the allocation obtained by $\pi' = (1, 2, \ldots, n, k, k, \ldots, k) \in \pickseq$.
In the first round, every agent $i \in N$ selects $g_i$.
In the subsequent rounds, agent $k$ selects all of $g_{n+1}, g_{n+2}, \ldots, g_m$.
Agent $k$ receives a utility of $1$ from picking all $n$ goods valuable to her, and every other agent~$i$ receives a utility of $1$ from $g_i$, so the egalitarian welfare is~$1$.
This shows that
\begin{align*}
    \frac{\EW(\pi',\mathcal{I})}{\EW(\pi,\mathcal{I})} = \frac{1}{1/n} = n,
\end{align*}
proving the lower bound when $m \geq 2n$.
\end{proof}

\subsection{Price Relative to Recursively Balanced Picking Sequences}
\label{subsec:ew_recbalseq}

Next, we consider the egalitarian price with respect to all \emph{recursively balanced} picking sequences with the same first-round prefix.
Again, we find that the egalitarian price is the same regardless of the picking sequence.

\begin{theorem}
\label{thm:ew_recbalseq}
For any $n \geq 2$, $m \geq n$, and $\pi \in \recbalseq$,
\begin{align*}
    \sup_{\mathcal{I} \in \mathcal{I}_{n, m}} \max_{\pi' \in \recbalseq} \frac{\EW(\pi',\mathcal{I})}{\EW(\pi,\mathcal{I})} &= \min\{\lceil m/n \rceil, \lfloor \log_2 n \rfloor + 1\}.
\end{align*}
\end{theorem}

\begin{proof}
We begin by proving the upper bound.
Let $\pi, \pi' \in \recbalseq$ and $\mathcal{I} \in \mathcal{I}_{n,m}$.

First, we show that when $n \leq m \leq n(\lfloor \log_2 n \rfloor+1)$, an upper bound is $\min\{\lceil m/n \rceil, \lfloor \log_2 n \rfloor + 1\} = \lceil m/n \rceil$.
Fix $k \in N$, and let $g$ be the good picked by agent $k$ in the first round of $\pi$.
Note that this good is also picked by agent $k$ in the first round of $\pi'$, since both picking sequences have the same prefix for the first round.
Agent $k$ receives a utility of at least $u_k(g)$ in the allocation obtained by $\pi$.
In $\pi'$, agent $k$ receives at most $\lceil m/n \rceil$ goods since $\pi'$ is recursively balanced, and the utility of each good is at most $u_k(g)$, so agent $k$ receives a utility of at most $\lceil m/n \rceil \cdot u_k(g)$.
Therefore, the ratio of agent $k$'s utility in the allocation obtained by $\pi'$ to that in the allocation obtained by $\pi$ is at most $\lceil m/n \rceil$.
Since $k \in N$ was arbitrarily chosen, we have
\begin{align*}
    \frac{\EW(\pi',\mathcal{I})}{\EW(\pi,\mathcal{I})} \leq \lceil m/n \rceil,
\end{align*}
proving the upper bound when $n \leq m \leq n(\lfloor \log_2 n \rfloor+1)$.

Next, we show that when $m > n(\lfloor \log_2 n \rfloor+1)$, an upper bound is $\min\{\lceil m/n \rceil, \lfloor \log_2 n \rfloor + 1\} = \lfloor \log_2 n \rfloor+1$.
For each $i \in N$ and positive integer~$r$, let $g_r^{i,\pi}$ and $g_r^{i,\pi'}$ be the goods selected by agent~$i$ in round~$r$ of $\pi$ and $\pi'$ respectively (for notational simplicity, if agent~$i$ does not select a good in round $r$, then we define the good $g_r^{i,\pi}$ or $g_r^{i,\pi'}$ to be a dummy good of zero utility).
Let $L = \lfloor \log_2 n \rfloor + 1$.
We first show the relationships between the goods via the following lemma.

\begin{restatable}{lemma}{lemEwRecbalseqUbLemma}
\label{lem:ew_recbalseq_ub_lemma}
For each $i \in N$ and positive integer $s$, we have
\[
    g_s^{i,\pi} \succeq_i g_{(s-1)L+1}^{i,\pi'}.
\]
\end{restatable}

\begin{proof}
We prove the statement by induction on $s$.
The base case of $s=1$ trivially holds since $g_1^{i,\pi} = g_1^{i,\pi'}$ for each $i \in N$.
Now, let $s \geq 2$ be a positive integer and assume that $g_{s'}^{i,\pi} \succeq_i g_{(s'-1)L+1}^{i,\pi'}$ holds for all $i \in N$ and $s' < s$.
We shall prove the statement for some fixed agent $k \in N$ and for $s$.
Let $s_1 = (s-1)L+1$, and assume for the sake of contradiction that $g_s^{k,\pi} \prec_k g_{s_1}^{k,\pi'}$.

Let $s_0 = (s-2)L+2$.
Since $L \geq 2$, we have $s_0 = (s-1)L-L+2 \leq (s-1)L < s_1$.
Also, $s_0 \geq (s-2)(1)+2 = s$.

Let $M' = \{g_t^{j,\pi'} \mid j \in N,\, s_0 \leq t \leq s_1\}$.
Note that the goods in $M'$ have not been selected yet at the start of round $s_0$ of $\pi'$.
We claim that every good in $M'$ has also not been selected yet at the start of round $s$ of $\pi$.
Indeed, if some agent $j \in N$ selected some $g \in M'$ in round $s' < s$ of $\pi$, then
\begin{align*}
    g &= g_{s'}^{j,\pi} \\
    &\succeq_j g_{s-1}^{j,\pi} \tag{since \protect{$s' \leq s-1$}} \\
    &\succeq_j g_{(s-2)L+1}^{j,\pi'} \tag{by inductive hypothesis} \\
    &= g_{s_0-1}^{j,\pi'} \tag{since \protect{$(s-2)L+1 = s_0-1$}} \\
    &\succ_j g, \tag{\protect{agent $j$ selected $g_{s_0-1}^{j,\pi'}$ over $g$ in round $s_0-1$ of $\pi'$}}
\end{align*}
a contradiction.
Therefore, every good in $M'$ is available at the start of round $s$ of $\pi$.

Define a directed graph $G$ such that the vertices represent the goods in $M'$, and there is an edge $g_{t}^{j,\pi'} \to g_{t-1}^{j',\pi'}$ if and only if $j'=j$ or $g_{t}^{j,\pi'} \in A_{j'}^\pi$.
For each $j \in N$, the vertex $g_{s_0}^{j,\pi'}$ has an outdegree of zero, and the vertex $g_t^{j,\pi'}$ has outdegree one or two for each $s_0 < t \leq s_1$. 

Consider the subgraph $G_k$ of $G$ induced by the vertex $g_{s_1}^{k,\pi'}$ and all of its successors (i.e., the vertices reachable from $g_{s_1}^{k,\pi'}$).
We claim that every good $g$ represented by a vertex in $G_k$ is selected by some agent $j' \neq k$ in round~$s$ of~$\pi$ before agent $k$'s turn in the same round.

To prove the claim, we first define, for each good $g$ represented by a vertex in $G_k$, its \emph{diversity index} $d$ to be the smallest number for which there exists a path from $g_{s_1}^{k,\pi'}$ to~$g$ such that the goods on the path belong to a total of at most $d$ agents (in the allocation produced by~$\pi'$).
We prove the claim by induction on $d$, beginning with the base case of $d=1$, i.e., the goods $g_t^{k,\pi'}$ for $s_0 \leq t \leq s_1$.
For the base case, we have the chain of relations $g_{s_0}^{k,\pi'} \succ_k g_{s_0+1}^{k,\pi'} \succ_k \cdots \succ_k g_{s_1}^{k,\pi'} \succ_k g_s^{k,\pi}$, where the last relation is by the assumption for the sake of contradiction, and the other relations are due to agent~$k$'s order of picking the goods in $\pi'$.
Since all the goods $g_t^{k,\pi'}$ for $s_0\le t\le s_1$ are available at the start of round $s$ of $\pi$ but agent~$k$ selected $g_s^{k,\pi}$ in round $s$ of $\pi$ instead, it must be the case that $g_t^{k,\pi'}$ are all selected in round~$s$ of $\pi$ by other agents who appear before agent $k$ in round $s$ of $\pi$.
This completes the proof of the base case.

For the inductive step, assume that there exists some positive integer $d$ such that every good with diversity index at most $d$ is selected by some agent in round $s$ of $\pi$ before agent $k$'s turn in the same round.
Consider a good $g_t^{j,\pi'}$ with diversity index $d+1$.
Then, there exist $t' \geq t$ and $j' \in N \setminus \{j\}$ such that
the edge $g_{t'+1}^{j',\pi'} \to g_{t'}^{j,\pi'}$ is in $G_k$, where $g_{t'+1}^{j',\pi'}$ has diversity index at most $d$.
We have the chain of relations $g_{s_0}^{j,\pi'} \succ_j g_{s_0+1}^{j,\pi'} \succ_j \cdots \succ_j g_{t'}^{j,\pi'} \succ_j g_{t'+1}^{j',\pi'}$ due to agent $j$'s order of picking the goods in $\pi'$.
However, the edge $g_{t'+1}^{j',\pi'} \to g_{t'}^{j,\pi'}$ implies that $g_{t'+1}^{j',\pi'} \in A_j^\pi$, and the inductive hypothesis states that $g_{t'+1}^{j',\pi'}$ is selected by some agent in round $s$ of $\pi$ before agent $k$'s turn in the same round.
Therefore, agent $j$ is the agent who selected $g_{t'+1}^{j',\pi'}$ in round $s$ of~$\pi$; in particular, agent~$j$ appears before agent~$k$ in this round.
However, the goods $g_{s_0}^{j,\pi'}, \ldots, g_{t'}^{j,\pi'}$ are available at the start of round $s$ of $\pi$ since they are in $M'$, but are not available to agent $j$ in round $s$ of $\pi$ (else, agent $j$ would have selected one of these goods instead of $g_{t'+1}^{j',\pi'}$ in round $s$ of $\pi$).
This means that the goods $g_{s_0}^{j,\pi'}, \ldots, g_{t'}^{j,\pi'}$, including $g_t^{j,\pi'}$, are all selected in round $s$ of $\pi$ by other agents who appear before agent $j$, who in turn appears before agent $k$, in round $s$ of $\pi$.
This completes the induction.
Therefore, every good represented by a vertex in $G_k$ is selected by some agent in round~$s$ of $\pi$ before agent~$k$'s turn in the same round.
Moreover, we also see from the proof that every good $g_t^{i,\pi'}$ represented by a vertex in $G_k$ is selected in round $s$ of~$\pi$ by some agent who is not $i$.
This means that every vertex in $G_k$ which has outdegree at least one has exactly two direct successors.

We shall next prove that $G_k$ is a tree.
Suppose on the contrary that $G_k$ is not a tree.
Then, there exists a vertex $g_{t}^{j,\pi'}$ in~$G_k$ such that there are at least two distinct paths $\mathcal{P}_1$ and $\mathcal{P}_2$ leading from $g_{s_1}^{k,\pi'}$ to it.
Without loss of generality, assume that there is no common vertex between $\mathcal{P}_1$ and $\mathcal{P}_2$ except for the start vertex~$g_{s_1}^{k,\pi'}$ and the end vertex~$g_{t}^{j,\pi'}$; otherwise, $g_{t}^{j,\pi'}$ can be replaced by the common vertex with the shortest distance from $g_{s_1}^{k,\pi'}$.

For each $c \in \{1, 2\}$, let $t_c \in \{t+1, t+2, \ldots, s_1\}$ be the smallest number such that $g_{t_c}^{j_c,\pi'}$ is a good in the path $\mathcal{P}_c$ and $j_c \neq j$.
We first observe that at least one of $t_1$ and $t_2$ exists and equals $t+1$; otherwise, both paths must go through $g_{t+1}^{j, \pi'}$, contradicting the assumption on the choice of $g_t^{j, \pi'}$. 
Without loss of generality, assume that $t_1 = t+1$.
We next show that $t_2$ must also exist.
Note that since $j_1 \neq j$, this means that $g_{t_1}^{j_1, \pi'} \in A_j^\pi$, that is, agent $j$ picks good $g_{t_1}^{j_1, \pi'}$ in $\pi$. By the earlier proven claim, agent $j$ picks this good in round $s$ before agent $k$'s turn in the same round.
Hence, $j \neq k$, so $t_2$ also exists.

Next, observe that $g_{t_1}^{j_1, \pi'}$ is an internal vertex of path~$\mathcal{P}_1$ and hence does not belong to path~$\mathcal{P}_2$.
Therefore, $g_{t_1}^{j_1, \pi'}$ and $g_{t_2}^{j_2, \pi'}$ are two distinct goods. 
On the other hand, by definition of $t_c$, there is an edge from $g_{t_c}^{j_c,\pi'}$ to $g_{t_c-1}^{j,\pi'}$ in $\mathcal{P}_c$ for each $c \in \{1, 2\}$.
However, this implies that agent $j$ selected two distinct goods in round $s$ of $\pi$, contradicting the assumption that $\pi$ is recursively balanced.
Therefore, $G_k$ is a tree.

Since every non-leaf vertex in $G_k$ has exactly two children, $G_k$ is a perfect binary tree of height $s_1-s_0$.
By counting, there are $2^{s_1-s_0+1}-1$ vertices in $G_k$.
Since the goods represented by the vertices in $G_k$ are selected by agents who appear before agent $k$ in round $s$ of $\pi$, there are at most $n-1$ vertices in $G_k$.
Hence, $2^{s_1-s_0+1}-1 \leq n-1$, which means that $s_1 \leq s_0 - 1 + \log_2 n$.
Since $s_1$ is an integer, we have
\begin{align*}
    s_1 &\leq s_0 - 1 + \lfloor \log_2 n \rfloor \\
    &= ((s-2)L + 2) - 1 + \lfloor \log_2 n \rfloor \\
    &= (s-1)L,
\end{align*}
contradicting our definition that $s_1 = (s-1)L + 1$.
This completes the inductive step of \Cref{lem:ew_recbalseq_ub_lemma}.
\end{proof}

We continue the proof of (the upper bound of) \cref{thm:ew_recbalseq}.
Fix $k \in N$.
We have
\begin{align*}
    u_k\left(A_k^{\pi'}\right) 
    &= \sum_{r=1}^{\lceil m/n \rceil} u_k\left(g_r^{k,\pi'}\right) \\
    &= \sum_{r=1}^{L\lceil m/n \rceil} u_k\left(g_r^{k,\pi'}\right) \tag{\protect{$u_k\left(g_r^{k,\pi'}\right) = 0$} for $r > \lceil m/n \rceil$} \\
    &= \sum_{s=1}^{\lceil m/n \rceil} \sum_{r=(s-1)L+1}^{sL} u_k\left(g_r^{k,\pi'}\right) \\
    &\leq \sum_{s=1}^{\lceil m/n \rceil} \sum_{r=(s-1)L+1}^{sL} u_k\left(g_{(s-1)L+1}^{k,\pi'}\right) \\
    &= \sum_{s=1}^{\lceil m/n \rceil} L \cdot u_k\left(g_{(s-1)L+1}^{k,\pi'}\right) \\
    &\leq \sum_{s=1}^{\lceil m/n \rceil} L \cdot u_k\left(g_s^{k,\pi}\right) \tag{by \cref{lem:ew_recbalseq_ub_lemma}} \\
    &= L \cdot u_k\left(A_k^\pi\right).
\end{align*}
This shows that in the allocation obtained with $\pi'$, agent $k$ receives a utility of at most $L$ times her utility in the allocation obtained with $\pi$.
Since $k \in N$ was arbitrarily chosen, we have
\begin{align*}
    \frac{\EW(\pi',\mathcal{I})}{\EW(\pi,\mathcal{I})} \leq L = \lfloor \log_2 n \rfloor + 1,
\end{align*}
proving the upper bound when $m > n(\lfloor \log_2 n \rfloor+1)$.

We continue by proving the lower bound.
When $m = n$, the allocation obtained by $\pi$ is identical to that obtained by any $\pi' \in \recbalseq$, so the ratio between the egalitarian welfare of the two allocations is $1 = \min\{\lceil m/n \rceil, \lfloor \log_2 n \rfloor + 1\}$.

Next, we show that when $n+1 \leq m \leq 2n-1$, a lower bound is $2 = \min\{\lceil m/n \rceil, \lfloor \log_2 n \rfloor + 1\}$.
Let $\pi \in \recbalseq$, and let $k \in N$ be an agent who does not get to pick in the second round---such an agent must exist since $n+1 \leq m \leq 2n-1$ and so the second round is incomplete.
Consider an instance~$\mathcal{I} \in \mathcal{I}_{n,m}$ with the following utilities:
\begin{itemize}
    \item For agent $k$, we have $u_k(g) = 1/m$ for $g \in M$.
    \item For each agent $i \in N \setminus \{k\}$, we have $u_i(g_i) = 1$.
    \item $u_i(g_j) = 0$ for all other pairs $(i, j)$.
\end{itemize}
Agents break ties in favor of lower-index goods.

Consider the allocation obtained by $\pi$.
Agent $k$ only selects $g_k$ and receives a utility of $1/m$, while every other agent~$i$ receives a utility of $1$ from $g_i$, so the egalitarian welfare is $1/m$.
Now, consider the allocation obtained by any $\pi' \in \recbalseq$ where agent $k$ picks first in the second round.
Agent $k$ selects $g_k$ and $g_{n+1}$, receiving a utility of $2/m$, while every other agent $i$ receives a utility of $1$ from $g_i$, so the egalitarian welfare is $2/m$.
This shows that
\begin{align*}
    \frac{\EW(\pi',\mathcal{I})}{\EW(\pi,\mathcal{I})} = \frac{2/m}{1/m} = 2,
\end{align*}
proving the lower bound when $n+1 \leq m \leq 2n-1$.

Finally, we show that when $m \geq 2n$, a lower bound is $\min\{\lceil m/n \rceil, \lfloor \log_2 n \rfloor + 1\}$.
Let $\pi \in \recbalseq$, and define $i_k = a_{n+k}$ for each $k \in \{1, \ldots, n\}$.
This means that $\pi = (1, \ldots, n \mid i_1, \ldots, i_n \mid \cdots)$.
Consider an instance~$\mathcal{I} \in \mathcal{I}_{n,m}$ with the following utilities, where $0 < \epsilon <1/((2n-1)m)$:
\begin{itemize}
    \item For agent $i_n$, we have $u_{i_n}(g_j) = 1/(2n-1)$ for $j \in \{1, \ldots, 2n-1\}$.
    \item For each agent $i_k \in N \setminus \{i_n\}$, we have $u_{i_k}(g_j) = \epsilon$ if $g_j \in M_k \coloneqq \{g_{n+1}, \ldots, g_{n+k}\} \cup \{g_{2n}, \ldots, g_m\}$.
    \item For each agent $i_k \in N \setminus \{i_n\}$, we have $u_{i_k}(g_{i_k}) = 1 - |M_k|\epsilon$.
    \item $u_i(g_j) = 0$ for all other pairs $(i, j)$.
\end{itemize}
Agents break ties in favor of lower-index goods.
Note that $1 - |M_k|\epsilon > 1 - m\epsilon > 1 - 1/(2n-1) > 1/(2n-1) > \epsilon$.

Consider the allocation obtained by $\pi$.
In the first round, every agent $i_k \in N$ selects $g_{i_k}$.
In every subsequent round, every agent $i_k \in N \setminus \{i_n\}$ selects a good with utility $\epsilon$, while agent $i_n$ selects a good with zero utility.
Note that every agent $i_k$ receives at least one valuable good, namely, $g_{i_k}$ in the first round, and agent $i_n$ receives exactly one valuable good.
Hence, every agent $i_k \in N \setminus \{i_n\}$ receives a utility of at least $1 - |M_k|\epsilon$, while agent $i_n$ gets a utility of exactly $1/(2n-1)$.
Since $1 - |M_k|\epsilon >  1/(2n-1)$, the egalitarian welfare is $1/(2n-1)$.

Now, define $\pi' \in \recbalseq$ by $a'_{(r-1)n+k} = i_{n-k+1}$ for all $r \geq 2$ and $k \in \{1, \ldots, n\}$. This means that $\pi' = (1, \ldots, n \mid i_n, \ldots, i_1 \mid \cdots \mid i_n, \ldots, i_1 \mid i_n, \ldots, i_{\lceil m/n \rceil n - m + 1})$.
Consider the allocation obtained by $\pi'$.
In the first round, every agent $i_k \in N$ selects $g_{i_k}$.
Let $M' = \{g_{n+1}, \ldots, g_{2n-1}\}$ be the goods that are valuable to agent $i_n$, excluding $g_1, \ldots, g_n$.
We first show, via the following lemma, that there is at least one good available in $M'$ at the start of round $r$ for each $2 \leq r \leq \lfloor \log_2 n \rfloor +1$.
Since agent $i_n$ always picks first in the second round onwards, this gives a lower bound on the number of valuable goods obtained by agent $i_n$ under picking sequence $\pi'$.

\begin{restatable}{lemma}{lemEwRecbalseqLbLemma}
\label{lem:ew_recbalseq_lb_lemma}
Let $2 \leq r \leq \lfloor \log_2 n \rfloor + 1$.
At the start of round $r$ of $\pi'$, if some good $g_j$ is available for some $j \geq 2n$, then there are exactly $\lfloor n/2^{r-2} \rfloor-1$ goods available in $M'$, namely, $g_{2n-\lfloor n/2^{r-2} \rfloor+1}, \ldots, g_{2n-1}$.
\end{restatable}

\begin{proof}
We prove this by induction on $r$.
For the base case of $r=2$, every agent selected goods in $\{g_1, \ldots, g_n\}$ in the first round, so $g_{2n}$ is available, and each of the $n-1 = \lfloor n/2^0 \rfloor-1$ goods in $M'$ is available.

Let $2 \leq r \leq \lfloor \log_2 n \rfloor$, and suppose that the statement is true for~$r$ as the inductive hypothesis; we shall prove the statement for $r+1$.
Assume that some good $g_j$ is available for some $j\ge 2n$ at the start of round $r+1$, and therefore also at the start of round~$r$.
By the inductive hypothesis, the goods $g_{2n-\lfloor n/2^{r-2} \rfloor+1}, \ldots, g_{2n-1}$ are available at the start of round~$r$.
Consider the goods selected in round~$r$.
Since the picking sequence for round $r$ is $i_n, i_{n-1}, \ldots, i_1$, starting from $p=1$ going upwards, agent $i_{n-p+1}$ selects her most preferred good (according to $\succ_{i_{n-p+1}}$).
Note that $\lfloor n/2^{r-2} \rfloor-1 \geq 1$, so agent $i_n$ selects a (valuable) good in $M'$, namely, $g_{2n-\lfloor n/2^{r-2} \rfloor+1}$.
For the subsequent agents, agent $i_{n-p+1}$ selects the good $g_{2n-\lfloor n/2^{r-2} \rfloor+p}$ provided that this good is in $M_{n-p+1}$, i.e., $2n-\lfloor n/2^{r-2} \rfloor+p \leq n+(n-p+1)$; the remaining agents select goods $g_j$ for $j \geq 2n$ (we may assume that such goods are available for every remaining agent; otherwise, no good $g_j$ for $j\ge 2n$ is available at the start of round $r+1$, contradicting our assumption).
The inequality can be solved as $p \leq (Q+1)/2$, where $n = 2^{r-2}\cdot Q + R$ for some integers $Q \ge 2$ and $0 \leq R < 2^{r-2}$.
Therefore, $\lfloor (Q+1)/2 \rfloor$ goods in $M'$ are selected in round $r$.
Since there were $\lfloor n/2^{r-2} \rfloor-1$ goods available in $M'$ at the start of the round, and $\lfloor (Q+1)/2 \rfloor$ goods in $M'$ are selected in this round, the number of goods remaining from $M'$ at the end of the round is
\begin{align*}
    &\left( \left\lfloor \frac{n}{2^{r-2}} \right\rfloor-1 \right) - \left\lfloor \frac{Q+1}{2} \right\rfloor \\ 
    &= \left( Q-1 \right) - \left\lfloor \frac{Q+1}{2} \right\rfloor \\
    &= \left\lceil \frac{2Q-2}{2} - \frac{Q+1}{2} \right\rceil \\
    &= \left\lceil \frac{Q-1}{2} \right\rceil - 1 \\
    &= \left\lfloor \frac{Q}{2} \right\rfloor - 1 \\
    &= \left\lfloor \frac{Q}{2} + \frac{R}{2^{r-1}} \right\rfloor - 1 \tag{since \protect{$R/2^{r-1} < 1/2$}} \\
    &= \left\lfloor \frac{2^{r-2} \cdot Q + R}{2^{r-1}} \right\rfloor - 1 \\
    &= \left\lfloor \frac{n}{2^{r-1}} \right\rfloor - 1.
\end{align*}
In other words, these $\lfloor n/2^{r-1} \rfloor-1$ goods were not taken by any agent in round $r$.
Therefore, at the start of round $r+1$, there are $\lfloor n/2^{r-1} \rfloor-1$ goods in $M'$ available, namely, $g_{2n-\lfloor n/2^{r-1} \rfloor+1}, \ldots, g_{2n-1}$.
This completes the proof of the inductive step.
\end{proof}

We can now continue the proof of \cref{thm:ew_recbalseq} by finding the lower bound when $m \geq 2n$.
Consider round $r$ where $2 \leq r \leq \min\{\lceil m/n \rceil, \lfloor \log_2 n \rfloor + 1\}$.
Suppose first that some good $g_j$ is available for some $j \geq 2n$ at the start of round~$r$.
Since $r \leq \lfloor \log_2 n \rfloor + 1$, there is at least one good in $M'$ by \cref{lem:ew_recbalseq_lb_lemma}, and agent $i_n$ selects a valuable good in $M'$.
Otherwise, $g_j$ is not available for any $j \geq 2n$, and so the only goods available are in $M'$.
Since $r \leq \lceil m/n \rceil$, there is still some good available at the start of round~$r$, which is in $M'$, so agent $i_n$ again selects a valuable good in $M'$.
Hence, for each round $r$ where $2 \leq r \leq \min\{\lceil m/n \rceil, \lfloor \log_2 n \rfloor + 1\}$, agent $i_n$ receives a valuable good.
Since agent $i_n$ also receives a valuable good in the first round, she receives a total of at least $\min\{\lceil m/n \rceil, \lfloor \log_2 n \rfloor + 1\}$ valuable goods, with a total utility of at least $\min\{\lceil m/n \rceil, \lfloor \log_2 n \rfloor + 1\}/(2n-1)$.
Every other agent receives a utility of at least $1 - |M_k|\epsilon > 1 - 1/(2n-1) = (2n-2)/(2n-1)$.
Since $n \geq 2$, we have $2n - 2 \geq n \ge \log_2 n + 1 \geq \min\{\lceil m/n \rceil, \lfloor \log_2 n \rfloor + 1\}$.
Hence, the egalitarian welfare under $\pi'$ is at least $\min\{\lceil m/n \rceil, \lfloor \log_2 n \rfloor + 1\}/(2n-1)$.
It follows that
\begin{align*}
    \frac{\EW(\pi',\mathcal{I})}{\EW(\pi,\mathcal{I})} &\ge \frac{\min\{\lceil m/n \rceil, \lfloor \log_2 n \rfloor + 1\}/(2n-1)}{1/(2n-1)} \\
    &= \min\{\lceil m/n \rceil, \lfloor \log_2 n \rfloor + 1\},
\end{align*}
proving the lower bound when $m \geq 2n$.
\end{proof}

\section{Maximin Share}
\label{sec:mms}

In this section, we compare recursively balanced picking sequences using MMS approximation.
Interestingly, for most sequences, the MMS guarantee depends only on the picking sequence of agent~$n$ (who picks last in the first round), regardless of the picking sequences of other agents.
Intuitively speaking, the disadvantage incurred by agent $n$ in the first round is so significant that letting her pick first in every subsequent round does not sufficiently compensate for it---she will still have the lowest MMS guarantee among all agents.

\subsection{Maximin Share Guarantees for Individual Agents}
\label{subsec:mms_agents}

Before we state our main results in \Cref{sec:MMS-picking}, we make a few observations regarding the maximin guarantees for each agent.
While it should not come as a surprise that the agent who picks last in the first round has a lower MMS guarantee than other agents, the extent of the decrease is noteworthy.
In this section, we find an (almost tight) MMS guarantee for each agent depending only on her positions in the picking sequence.
Furthermore, we obtain simpler bounds which will be useful in proving our main results.

Our main lemma describes the MMS guarantee for each agent as follows.

\begin{lemma}
\label{lem:mms_agent_lb}
Let $n \geq 2$, $m \geq n$, $\pi \in \recbalseq$, and $i \in N$.
Let agent $i$'s picking sequence be $\pi_i = (t_1, \ldots, t_R)$ and define $t_{R+1} = m+1$.
Then, in the allocation obtained with $\pi$, agent $i$ is guaranteed to get utility at least
\[
    \min_{r \in \{2, \ldots, R+1\}} \, \frac{(r-1) \cdot (n+1-i)}{t_r-i}
\]
times her MMS.
\end{lemma}

\begin{proof}
Take any instance~$\mathcal{I} \in \mathcal{I}_{n, m}$.
Without loss of generality, we assume that $g_1 \succ_i g_2 \succ_i \cdots \succ_i g_m$,
and let $M_{\geq i} = \{g_i, g_{i+1}, \ldots, g_m\}$.
Suppose that agent $i$'s MMS is non-zero, as the lemma holds trivially otherwise.

We first show that agent $i$'s MMS is at most $u_i(M_{\geq i})/(n+1-i)$.
Consider any MMS partition $\{P_1, \ldots, P_n\}$ of agent~$i$.
Without loss of generality, assume that each $P_k$ is non-empty, and $P_1, \ldots, P_n$ are sorted by the most preferred good in $P_k$ according to $\succ_i$, in decreasing order by $\succ_i$.
That is, $\min \{j : g_j \in P_1\} < \min \{j : g_j \in P_2\} < \cdots < \min \{j : g_j \in P_n\}$.
Note that $\min \{j : g_j \in P_k\} \geq k$ for every $k\in N$.
This implies that $g_1, \ldots, g_{i-1} \not\in \bigcup_{k=i}^n P_k$, and so $\bigcup_{k=i}^n P_k \subseteq M_{\geq i}$. It follows that
\begin{equation}
\label{eq:mms_ub}
\begin{aligned}
    \MMS_{i} &\leq \min_{k \in \{i, \ldots, n\}} u_{i}(P_k) \\
    &\leq \frac{1}{n+1-i} \cdot u_i\left(\textstyle\bigcup_{k=i}^n P_k\right) \\
    &\leq \frac{u_i(M_{\geq i})}{n+1-i}.
\end{aligned}    
\end{equation}

We continue by finding an upper bound on $u_i(M_{\geq i})$.
In her $r$-th pick for each $r \in \{1, \ldots, R\}$, agent $i$ can choose a good of value at least $u_i(g_{t_r})$.
Hence, agent $i$'s total utility under picking sequence~$\pi$ is 
\begin{equation}
\label{eq:bundle_lb}
    u_i(A_i) \geq \sum_{r=1}^R u_i(g_{t_r}).
\end{equation}
Furthermore, we have
\begin{equation}
\label{eq:goods_ub}
\begin{aligned}
    u_i(M_{\geq i}) &= \sum_{j=i}^{m} u_i(g_j) \\
    &= \sum_{j=t_1}^{t_{R+1}-1} u_i(g_j) \\
    &= \sum_{r=1}^R \sum_{j=t_r}^{t_{r+1}-1} u_i(g_j) \\
    &\leq \sum_{r=1}^R \sum_{j=t_r}^{t_{r+1}-1} u_i(g_{t_r}) \\
    &= \sum_{r=1}^R (t_{r+1}-t_r)\cdot u_i(g_{t_r}).
\end{aligned}
\end{equation}
From \eqref{eq:mms_ub} and the assumption that agent $i$'s MMS is non-zero, we have $u_i(M_{\geq i}) > 0$.
Then, combining \eqref{eq:bundle_lb} and~\eqref{eq:goods_ub} gives
\begin{align}
\label{eq:mms_agent_lb}
    \frac{u_i(A_i)}{u_i(M_{\geq i})} \geq \frac{\sum_{r=1}^{R} u_i(g_{t_r})}{\sum_{r=1}^{R} (t_{r+1}-t_r)\cdot u_i(g_{t_r})}.
\end{align}

Let $y$ be the value of the fraction on the right-hand side of~\eqref{eq:mms_agent_lb}.
Let $x_R = u_i(g_{t_R})$, and for each $r \in \{1, \ldots, R-1\}$, let $x_r = u_i(g_{t_r}) - u_i(g_{t_{r+1}})$.
By telescoping sum, $u_i(g_{t_r}) = x_r + x_{r+1} + \cdots + x_R$.
Therefore, 
\begin{align*}
    y &= \frac{\sum_{r=1}^{R} u_i(g_{t_r})}{\sum_{r=1}^{R} (t_{r+1}-t_r)\cdot u_i(g_{t_r})} \\
    &= \frac{\sum_{r=1}^R \sum_{s=r}^R x_s}{\sum_{r=1}^R \sum_{s=r}^R (t_{r+1}-t_r) x_s} \\
    &= \frac{x_1 + 2x_2 + \cdots + Rx_R}{(t_2-t_1)x_1 + (t_3-t_1)x_2 + \cdots + (t_{R+1}-t_1)x_R}.
\end{align*}
Observe that $y$ is a weighted average of $1/(t_2-t_1), 2/(t_3-t_1), \ldots, R/(t_{R+1}-t_1)$, with the weights being $(t_2-t_1)x_1, (t_3-t_1)x_2, \ldots, (t_{R+1}-t_1)x_R$ respectively.
(Note that the weights may \emph{not} sum up to~$1$, so we have to normalize by dividing with the sum of the weights.)
Hence, $y$ cannot have value smaller than the minimum of its data elements $1/(t_2-t_1), 2/(t_3-t_1), \ldots, R/(t_{R+1}-t_1)$.
Combining this observation with \eqref{eq:mms_ub} and~\eqref{eq:mms_agent_lb} gives
\begin{align*}
    \frac{u_i(A_i)}{\MMS_i} &\geq \frac{u_i(A_i)}{u_i(M_{\geq i})} \cdot (n+1-i) \tag{from \eqref{eq:mms_ub}} \\
    &\geq y \cdot (n+1-i) \tag{from \eqref{eq:mms_agent_lb}} \\
    &\geq \left(\min_{r \in \{2, \ldots, R+1\}} \frac{r-1}{t_{r}-t_1}\right) \cdot (n+1-i) \\
    &= \min_{r \in \{2, \ldots, R+1\}} \frac{(r-1) \cdot (n+1-i)}{t_{r}-i},
\end{align*}
as desired.
\end{proof}

We next show via the following upper bound that the MMS guarantee in \cref{lem:mms_agent_lb} is always tight for agent $n$ and sometimes tight for every other agent.

\begin{lemma}
\label{lem:mms_agent_ub}
Let $n \geq 2$, $m \geq n$, $\pi \in \recbalseq$, and $i \in N$.
Let agent $i$'s picking sequence be $\pi_i = (t_1, \ldots, t_R)$ and define $t_{R+1} = m+1$.
There exists an instance such that under the picking sequence $\pi$, agent $i$ has a positive MMS and gets no more than
\[
    \min_{r \in \{2, \ldots, R+1\}} \, \frac{r-1}{\lfloor (t_r-i)/(n+1-i) \rfloor}
\]
times her MMS.
\end{lemma}

\begin{proof}
Choose $s \in \{2, \ldots, R+1\}$ which minimizes $(r-1)/\lfloor (t_r-i)/(n+1-i) \rfloor$.
Consider an instance with the following utilities:
\begin{itemize}
    \item $u_i(g_j) = m$ for each $j \in \{1, \ldots, i-1\}$.
    \item $u_i(g_j) = 1$ for each $j \in \{i, \ldots, t_s-1\}$.
    Note that this set is non-empty, since $s\ge 2$ and so $t_s > n \ge i$.
    \item For each agent $k \in N \setminus \{i\}$, we have $u_k(g_k) = 1$.
    \item $u_\ell(g_j) = 0$ for all other pairs $(\ell, j)$.
\end{itemize}
Agents break ties in favor of lower-index goods.

Observe that agent $i$'s MMS partition is $\{\{g_1\}, \{g_2\}, \ldots, \{g_{i-1}\},  P_i, P_{i+1}, \ldots, P_n\}$, where the goods $g_i, g_{i+1}, \ldots, g_{t_s-1}$ are distributed as evenly as possible among the sets $P_i, P_{i+1}, \ldots, P_n$.
Thus, agent~$i$'s MMS is exactly $\lfloor (t_s-i)/(n+1-i) \rfloor$, which is at least $1$ since $t_s \geq t_2 \geq n+1$.
Now, each good $g_j$ is taken by some agent during the $j$-th overall pick of $\pi$. 
On the other hand, since $t_s$ is the overall position of agent $i$'s $s$-th pick, agent $i$ can only pick $s-1$ goods out of $g_1, \ldots, g_{t_s-1}$. Furthermore, $g_1, \ldots, g_{i-1}$ are taken by agents $1, \ldots, i-1$ respectively. Hence, agent $i$ gets a total utility of $s-1$, which is exactly
\[
    \frac{s-1}{\lfloor (t_s-i)/(n+1-i) \rfloor}
\]
times her MMS.
\end{proof}

The remaining lemmas in this section provide simpler upper and lower bounds for the MMS guarantees of different agents, and will be used in proving our main results in \Cref{sec:MMS-picking}.

\begin{lemma}
\label{lem:mms_agent_lb_large_m}
Let $n \geq 2$, $m \geq n$, $\pi \in \recbalseq$, and $i \in N$.
In the allocation obtained with $\pi$, agent $i$ is guaranteed to get utility at least $(n+1-i)/(2n-i)$ times her MMS.
\end{lemma}

\begin{proof}
Let $t_1, \ldots, t_{R+1}$ be defined as in \cref{lem:mms_agent_lb}, that is, agent $i$'s picking sequence under $\pi$ is $\pi_i = (t_1, \ldots, t_R)$ and $t_{R+1} = m+1$.
Since $\pi$ is recursively balanced, $t_r \leq rn$ for each $r \in \{2, \ldots, R\}$. 
Furthermore, $R \geq \lfloor m/n \rfloor$, and so $(R+1)n \geq (\lfloor m/n \rfloor+1)n > (m/n)n = m$. 
Hence, we also have $t_{R+1} = m+1 \leq (R+1)n$.
For each $r \in \{2, \ldots, R+1\}$, it holds that
\begin{align*}
    \frac{r-1}{t_r-i} &\geq \frac{r-1}{rn-i} \tag{since $t_r \leq rn$} \\
    &= \frac{r-1}{rn-n+n-i} \\
    &= \frac{1}{n+\frac{n-i}{r-1}} \\
    &\geq \frac{1}{2n-i}. \tag{since $r \geq 2$}
\end{align*}
It follows that
\begin{align*}
    \min_{r \in \{2, \ldots, R+1\}} \frac{(r-1) \cdot (n+1-i)}{t_r-i} \geq \frac{n+1-i}{2n-i}.
\end{align*}
Hence, by \cref{lem:mms_agent_lb}, agent $i$ is always guaranteed to get utility at least $(n+1-i)/(2n-i)$ times her MMS.
\end{proof}

\begin{lemma}
\label{lem:mms_agent_lb_small_m}
Let $n \geq 2, m \geq n$, $\pi \in \recbalseq$, and $i \in N$.
In the allocation obtained with $\pi$, agent $i$ is guaranteed to get utility at least
\[
    \frac{1}{\lfloor (m+1-i)/(n+1-i) \rfloor}
\]
times her MMS.
\end{lemma}

\begin{proof}
Without loss of generality, we assume that $g_1 \succ_i g_2 \succ_i \cdots \succ_i g_m$.
We first modify the proof of \cref{lem:mms_agent_lb} to show that the MMS of agent $i$ is at most $u_i(g_i) \cdot \lfloor (m+1-i)/(n+1-i) \rfloor$.

Consider any MMS partition $\{P_1, \ldots, P_n\}$.
Without loss of generality, assume that each $P_k$ is non-empty, and $P_1, \ldots, P_n$ are sorted by the most preferred good in $P_k$ according to $\succ_i$, in decreasing order by $\succ_i$.
That is, $\min \{j : g_j \in P_1\} < \min \{j : g_j \in P_2\} < \cdots < \min \{j : g_j \in P_n\}$.
Note that $\min \{j : g_j \in P_k\} \geq k$ for any $k \in N$.
For any $g_j \in P_k$, it holds that $j \geq k$, and so $u_i(g_j) \leq u_i(g_k)$. 
Thus, $u_i(P_k) \leq u_i(g_k) \cdot |P_k|$. In particular, when $k \geq i$, we have $u_i(P_k) \leq u_i(g_i) \cdot |P_k|$.
Note that
\begin{align*}
    \min_{k \in \{i, \ldots, n\}} |P_k| &\leq \left\lfloor \frac{\left|P_i \cup \cdots \cup P_n \right|}{|\{i, \ldots, n\}|} \right\rfloor \\
    &\leq \left\lfloor \frac{m+1-i}{n+1-i} \right\rfloor. \tag{since $\bigcup_{k=i}^n P_k \subseteq \{g_i, \ldots, g_m\}$}
\end{align*}
Therefore, the MMS of agent $i$ is
\begin{align*}
    \MMS_i &\leq \min_{k \in \{i, \ldots, n\}} u_i(P_k) \\
    &\leq \min_{k \in \{i, \ldots, n\}} \left(u_i(g_i) \cdot |P_k|\right) \\
    &= u_i(g_i) \cdot \min_{k \in \{i, \ldots, n\}} |P_k| \\
    &\leq u_i(g_i) \cdot \left\lfloor \frac{m+1-i}{n+1-i} \right\rfloor.
\end{align*}
By definition of $\pi \in \recbalseq$, agent $i$'s bundle has value at least $u_i(g_i)$, which is at least $1/\lfloor (m+1-i)/(n+1-i) \rfloor$ times her MMS.
\end{proof}

\begin{lemma}
\label{lem:mms_agent_n-1_lb_regular}
Let $n \geq 2$, $m \geq n$, and $\pi \in \recbalseq$.
Suppose that $\pi$ satisfies at least one of the following conditions:
\begin{enumerate}
    \item There is only one round, that is, $m=n$.
    \item The second round is incomplete and has an odd number of turns.
    \item Agent $n-1$ gets to pick in the second round.
\end{enumerate}
Then, in the allocation obtained with $\pi$, agent $n-1$ is guaranteed to get at least
\[
    \alpha_{\max} = \min\left\{\frac{\lfloor m/n \rfloor}{\lfloor m/n \rfloor n - n + 1}, \frac{\lceil m/n \rceil}{m-n+1}\right\}
\]
times her MMS.
\end{lemma}

\begin{proof}
We first consider the case where $m = n$. 
Then, an MMS partition of agent $n-1$ is $\{\{g_1\}, \{g_2\}, \ldots, \{g_n\}\}$. Since agent $n-1$ is guaranteed to get one good, she gets at least her MMS.
On the other hand, note that $\alpha_{\max} = \min\{1, 1\} = 1$.

Next, consider the case where $n+1 \leq m \leq 2n-1$. In this case, $\alpha_{\max} = \min\{1,\, 2/(m-n+1)\} = 2/(m-n+1)$.
Furthermore, assume that $\pi$ satisfies either condition 2 or 3.

Suppose that $\pi$ satisfies condition 2.
Then, $m-n$ is odd.
By \cref{lem:mms_agent_lb_small_m}, agent $n-1$ is guaranteed to get at least $1/\lfloor (m-n+2)/2 \rfloor = 1/((m-n+1)/2) = \alpha_{\max}$ times her MMS.

Suppose now that $\pi$ satisfies condition 3.
Let $\pi_{n-1} = (t_1, \ldots, t_R)$ be the picking sequence of agent $n-1$, and let $t_{R+1} = m+1$.
Since $n+1 \leq m \leq 2n-1$, and agent $n-1$ picks in the second round, we have $R = 2$.
By \cref{lem:mms_agent_lb}, agent $n-1$ is guaranteed to get at least
\begin{align*}
    \min_{r \in \{2, 3\}} \frac{2(r-1)}{t_r-n+1} &= \min\left\{ \frac{2}{t_2-n+1}, \frac{4}{t_3-n+1} \right\} \\
    &\geq \min\left\{ \frac{2}{m-n+1}, \frac{4}{m-n+2} \right\} \tag{since $t_2 \leq m$ and $t_3 = m+1$} \\
    &= \frac{2}{m-n+1} \tag{since $m-n \geq 0$} \\
    &= \alpha_{\max}
\end{align*}
times her MMS.

Lastly, consider the case where $m \geq 2n$.
By \cref{lem:mms_agent_lb_large_m}, agent $i$ is guaranteed to get at least $(n+1-i)/(2n-i)$ times her MMS.
For $i \le n-1$, we have
\begin{align}
    \frac{n+1-i}{2n-i} &= \frac{1}{1+\frac{n-1}{n+1-i}} \notag \\
    &\geq \frac{1}{1+\frac{n-1}{2}} \tag{since $i \leq n-1$} \\
    &= \frac{1}{n - \frac{n-1}{2}} \notag \\
    &\geq \frac{1}{n - \frac{n-1}{\lfloor m/n \rfloor}} \tag{since $\lfloor m/n \rfloor \geq 2$} \\
    &= \frac{\lfloor m/n \rfloor}{\lfloor m/n \rfloor n - n + 1} \notag \\
    &\geq \alpha_{\max}. \label{eq:mms_agent_n-1_lb}
\end{align}
Hence, agent $n-1$ is guaranteed to get at least $\alpha_{\max}$ times her MMS.
\end{proof}

\begin{lemma}
\label{lem:mms_agent_1_to_n-2_lb}
Let $n \geq 2$, $m \geq n$, $\pi \in \recbalseq$, and $i \in N \setminus \{n-1,n\}$.
In the allocation obtained with $\pi$, agent $i$ is guaranteed to get at least 
\[
    \alpha_{\max} = \min\left\{\frac{\lfloor m/n \rfloor}{\lfloor m/n \rfloor n - n + 1}, \frac{\lceil m/n \rceil}{m-n+1}\right\}
\]
times her MMS.
\end{lemma}

\begin{proof}
Suppose first that $m \geq 2n$. 
By \cref{lem:mms_agent_lb_large_m}, agent~$i$ is guaranteed to get at least $(n+1-i)/(2n-i)$ times her MMS. Moreover, by \eqref{eq:mms_agent_n-1_lb}, $(n+1-i)/(2n-i) \geq \alpha_{\max}$ when $i \leq n-1$ and $m \geq 2n$. Hence, agent $i$ is guaranteed to get at least $\alpha_{\max}$ times her MMS.

We can now assume that $m \leq 2n-1$. 
By \cref{lem:mms_agent_lb_small_m}, agent $i$ is guaranteed to get at least $1/\lfloor (m+1-i)/(n+1-i) \rfloor$ times her MMS.
Hence, it suffices to prove that $1/\lfloor (m+1-i)/(n+1-i) \rfloor \geq \alpha_{\max}$.

Consider the case where $n+3 \leq m \leq 2n-1$. Then, $\alpha_{\max} = \min\{1,\, 2/(m-n+1)\} = 2/(m-n+1)$.
Note that 
\begin{align*}
    \frac{1}{\left\lfloor \frac{m+1-i}{n+1-i} \right\rfloor} &\geq \frac{n+1-i}{m+1-i} \\
    &= \frac{1}{\frac{m-n}{n+1-i}+1} \\
    &\geq \frac{1}{\frac{m-n}{3}+1} \tag{since $i \leq n-2$} \\
    &= \frac{3}{m-n+3}.
\end{align*}
Furthermore, $3(m-n+1) - 2(m-n+3) = m-n-3 \geq 0$ by the assumption on $m$.
Thus, $3/(m-n+3) \geq 2/(m-n+1) = \alpha_{\max}$, which completes the proof for this case.

Lastly, consider the case where $n \leq m \leq n+2$. Note that
\begin{align*}
    m+1-i &= (m-n)+(n+1-i) \\
    &\leq 2+(n+1-i) \tag{by assumption} \\
    &< 2(n+1-i). \tag{since $2 < n+1-i$}
\end{align*}
Thus, $\lfloor (m+1-i)/(n+1-i) \rfloor = 1$. On the other hand, 
\begin{align*}
    \alpha_{\max} &\leq \frac{\lfloor m/n \rfloor}{\lfloor m/n \rfloor n-n+1} \\
    &= \frac{1}{n-\frac{n-1}{\lfloor m/n \rfloor}} \\
    &\leq 1 \tag{since $\lfloor m/n \rfloor \geq 1$} \\
    &= \frac{1}{\lfloor (m+1-i)/(n+1-i) \rfloor}.
\end{align*}
This completes the proof.
\end{proof}

\begin{lemma}
\label{lem:mms_agent_n_ub}
Let $n \geq 2$, $m \geq n$, and $\pi \in \recbalseq$.
There exists an instance such that under the picking sequence $\pi$, agent $n$ has a positive MMS and gets no more than
\[
    \alpha_{\max} = \min\left\{\frac{\lfloor m/n \rfloor}{\lfloor m/n \rfloor n - n + 1}, \frac{\lceil m/n \rceil}{m-n+1}\right\}
\]
times her MMS.
\end{lemma}

\begin{proof}
Consider first the instance~$\mathcal{I}_1 \in \mathcal{I}_{n, m}$ with the following utilities:
\begin{itemize}
    \item For $i \in N \setminus \{n\}$, $u_i(g_i) = m$ and $u_i(g_j) = 0$ if $j \neq i$.
    \item $u_n(g_j) = m$ if $j < n$.
    \item $u_n(g_j) = 1$ if $n \leq j \leq \lfloor m/n \rfloor n$.
    \item $u_n(g_j) = 0$ if $j > \lfloor m/n \rfloor n$.
\end{itemize}
Agents break ties in favor of lower-index goods.

We observe that agent~$n$'s MMS is exactly $u_n(\{g_n, \ldots, g_m\}) = \allowbreak \lfloor m/n \rfloor n - n + 1 \geq 1$, as obtained by the partition $\{\{g_1\}, \{g_2\}, \ldots, \{g_{n-1}\}, \{g_n, g_{n+1}, \ldots, g_{m}\}\}$. 
In each round $r \in \{1, \ldots, \lfloor m/n \rfloor\}$, agent $n$ gets a good with utility $1$. 
If there is an incomplete last round $r = \lfloor m/n \rfloor + 1$, either agent $n$ does not get to pick at all, or she picks a good with utility $0$. In both cases, agent $n$'s total utility is $\lfloor m/n \rfloor$. 
Hence, agent~$n$ gets exactly $\lfloor m/n \rfloor/(\lfloor m/n \rfloor n - n + 1)$ times her MMS.

Now, consider a different instance~$\mathcal{I}_2 \in \mathcal{I}_{n, m}$ with the following utilities:
\begin{itemize}
    \item For $i \in N \setminus \{n\}$, $u_i(g_i) = m$ and $u_i(g_j) = 0$ if $j \neq i$.
    \item $u_n(g_j) = m$ if $j < n$.
    \item $u_n(g_j) = 1$ if $n \leq j \leq m$.
\end{itemize}
Again, agents break ties in favor of lower-index goods.

In this instance, agent $n$'s MMS is exactly $u_n(\{g_n, \ldots, g_m\}) = m - n + 1 \geq 1$, as obtained by the~partition $\{\{g_1\}, \{g_2\}, \ldots, \{g_{n-1}\}, \{g_n, g_{n+1}, \ldots, g_{m}\}\}$. Moreover, in each round, agent $n$ gets a good with utility at most $1$. 
Hence, agent $n$'s total utility is at most $\lceil m/n \rceil$.
It follows that agent~$n$ gets at most $\lceil m/n \rceil/(m - n + 1)$ times her MMS.

Therefore, for one of the instances $\mathcal{I}_1$ and $\mathcal{I}_2$, agent $n$ gets no more than \[
    \alpha_{\max} = \min\left\{\frac{\lfloor m/n \rfloor}{\lfloor m/n \rfloor n - n + 1}, \frac{\lceil m/n \rceil}{m-n+1}\right\}
\] times her MMS.
\end{proof}

\subsection{Maximin Share Guarantees for Recursively Balanced Picking Sequences}
\label{sec:MMS-picking}

We are now ready to state and prove our main results on MMS.
As mentioned earlier, for most sequences, the MMS guarantee depends only on the picking sequence of agent $n$.
However, there are some exceptions where agent $n-1$ has a lower MMS guarantee.
Hence, before we state the MMS guarantees, we distinguish between the general case and the exception by defining ``irregular'' picking sequences.

\begin{definition}
Let $n \geq 2$, $m \geq n$, and $\pi \in \recbalseq$.
A picking sequence $\pi$ is called \emph{irregular} if it has a second round that satisfies the following conditions:
\begin{enumerate}
    \item The second round has a positive even number of turns.
    \item Agent $n-1$ does not get to pick in the second round.
    \item Agent $n$ picks in the first $(m-n)/2$ turns of the second round.
\end{enumerate}
Otherwise, we say that $\pi$ is \emph{regular}.
\end{definition}

Observe that a picking sequence $\pi \in \recbalseq$ can be irregular only if $n+2 \leq m \leq 2n-1$. 
Indeed, if $m \leq n+1$, then the second round either does not exist or has an odd number of turns. 
If $m \geq 2n$, agent $n-1$ has to pick in the second round as $\pi$ is recursively balanced. 

We begin by characterizing the MMS guarantees for regular picking sequences.
Note that these guarantees depend only on the picking sequence $\pi_n$ of agent $n$.

\begin{restatable}{theorem}{mmsRegular}
\label{thm:mms_regular}
Let $n \geq 2$, $m \geq n$, and $\pi \in \recbalseq$. Suppose that $\pi$ is regular.
Let agent $n$'s picking sequence in $\pi$ be $\pi_n = (t_1, \ldots, t_R)$ where $t_1 = n$, and define $t_{R+1} = m+1$.
Let 
\[
    \alpha = \min_{r \in \{2, \ldots, R+1\}} \frac{r-1}{t_r-n}.
\]
Then, the following hold for the allocation obtained with~$\pi$:
\begin{enumerate}[(a)]
    \item Every agent is always guaranteed to get at least $\alpha$ times her MMS.
    \item There exists an instance~$\mathcal{I} \in \mathcal{I}_{n, m}$ where agent $n$ has a positive MMS and gets no more than $\alpha$ times her MMS.
\end{enumerate}
\end{restatable}

\begin{proof}
\cref{thm:mms_regular}(b) follows directly from \cref{lem:mms_agent_ub} with $i = n$.
Therefore, it remains to prove \cref{thm:mms_regular}(a).

By \cref{lem:mms_agent_lb}, agent $n$ is always guaranteed to get at least $\alpha$ times her MMS.
Hence, by \cref{lem:mms_agent_n_ub}, it holds that $\alpha \leq \alpha_{\max}$, where $\alpha_{\max}$ is the MMS guarantee of any agent $i \in N \setminus \{n-1, n\}$ according to \cref{lem:mms_agent_1_to_n-2_lb}.
This means that every agent $i \in N \setminus \{n-1, n\}$ is also guaranteed to get at least $\alpha$ times her MMS. 

It remains to prove the statement for agent $n-1$.
If $\pi$ satisfies at least one of the conditions in \cref{lem:mms_agent_n-1_lb_regular}, then agent $n-1$ is guaranteed to get at least $\alpha_{\max} \geq \alpha$ times her MMS.
Otherwise, the following three conditions must be satisfied:
\begin{enumerate}
    \item There is a second round, that is, $m \geq n+1$.
    \item The second round has an even number of turns.
    \item Agent $n-1$ does not get to pick in the second round.
\end{enumerate}
In particular, by condition~3, the second round must be incomplete.
Hence, $\pi$ satisfies conditions 1 and 2 of irregular picking sequences.
Since $\pi$ is regular, it must not satisfy condition 3, that is, agent $n$ does not pick in the first $(m-n)/2$ turns of the second round.
Let $t_2$ be defined as in the theorem statement: $t_2$ is the index of agent $n$'s second pick if she gets to pick in the second round, and $t_2 = m + 1$ otherwise.
In either case, $t_2 > n + (m-n)/2$.
Since the second round has an even number of turns, $(m-n)/2$ is an integer, and so $t_2 \geq n + (m-n)/2 + 1$.
This implies that $\alpha \leq 1/(t_2-n) \leq 1/((m-n)/2+1) = 2/(m-n+2)$.
On the other hand, by \cref{lem:mms_agent_lb_small_m}, agent $n-1$ is guaranteed to get at least $1/\lfloor (m-n+2)/2 \rfloor \geq 2/(m-n+2) \geq \alpha$ times her MMS, completing the proof.
\end{proof}

We next show that the MMS guarantee in \Cref{thm:mms_regular}(a) does not necessarily apply to irregular picking sequences.

\begin{example}
\label{ex:irregular-fail}
Let $n = 3$, $m = 5$, and $\pi = (1, 2, 3 \mid 3, 1)$.
Note that $\pi$ is irregular.
Consider the following instance:
\begin{center} 
    \begin{tabular}{c|ccccc}
        $g$      & $g_1$ & $g_2$ & $g_3$ & $g_4$ & $g_5$ \\
        \hline
        $u_1(g)$ & $6$ & $0$ & $0$ & $0$ & $0$ \\
        $u_2(g)$ & $2$ & $1$ & $1$ & $1$ & $1$ \\
        $u_3(g)$ & $6$ & $0$ & $0$ & $0$ & $0$ \\
    \end{tabular}
\end{center}
Agents break ties in favor of lower-index goods.

Note that agent~$2$ receives exactly one good.
Furthermore, this good cannot be $g_1$ since $g_1$ is picked by agent~$1$ in her first pick.
Hence, agent~$2$ gets a utility of $1$.
On the other hand, agent $2$ has an MMS  of $2$ by the partition $\{\{g_1\}, \{g_2, g_3\}, \{g_4, g_5\}\}$.
Thus, agent~$2$ gets exactly $1/2$ of her MMS.

We compute $\alpha$ as defined in \cref{thm:mms_regular}.
The picking sequence of agent $3$ in $\pi$ is $\pi_3 = (3, 4)$.
In particular, we have $R = 2$, $t_2 = 4$, and $t_3 = 6$, 
which gives $\alpha = \min_{r \in \{2, 3\}} (r-1)/(t_r-n) = \min\{1, 2/3\} = 2/3$.
Since agent $2$ only gets $1/2 < \alpha$ times her MMS, this shows that the guarantee in \cref{thm:mms_regular}(a) does not necessarily apply to irregular picking sequences.
\end{example}

In light of \Cref{ex:irregular-fail}, we separately determine the MMS guarantees for irregular picking sequences.

\begin{restatable}{theorem}{mmsIrregular}
\label{thm:mms_irregular}
Let $n \geq 2$, $m \geq n$, and $\pi \in \recbalseq$.
Suppose that $\pi$ is irregular.
Then, the following hold for the allocation obtained with~$\pi$:
\begin{enumerate}[(a)]
    \item Every agent is always guaranteed to get at least $2/(m-n+2)$ times her MMS.
    \item There exists an instance~$\mathcal{I} \in \mathcal{I}_{n, m}$ where agent $n-1$ has a positive MMS and gets no more than $2/(m-n+2)$ times her MMS.
\end{enumerate}
\end{restatable}

\begin{proof}
We first prove statement (a).
Since $\pi$ is irregular, we have $n+2 \leq m \leq 2n-1$.
By \cref{lem:mms_agent_1_to_n-2_lb}, each agent $i \in N \setminus \{n-1, n\}$ is guaranteed to get at least $\alpha_{\max} = \min\{1, 2/(m-n+1)\} > 2/(m-n+2)$ times her MMS.
Furthermore, by \cref{lem:mms_agent_lb_small_m}, agent $n-1$ is guaranteed to get at least $1/\lfloor (m-n+2)/2 \rfloor \geq 2/(m-n+2)$ times her MMS.
To prove the same for agent $n$, consider the picking sequence $\pi_n = (t_1, t_2)$ of agent $n$, and let $t_3 = m+1$. Since $\pi$ is irregular, $t_2 \leq n + (m-n)/2$.
By \cref{lem:mms_agent_lb}, agent $n$ is guaranteed to get at least $\min\{1/(t_2-n), 2/(t_3-n)\} \ge \min\{2/(m-n), 2/(m-n+1)\} > 2/(m-n+2)$ times her MMS. This completes the proof for statement (a).

We now prove statement (b).
Since $\pi$ is irregular, agent $n-1$ only picks once and $m-n$ is even.
Consider the picking sequence $\pi_{n-1} = (n-1)$ of agent $n-1$.
By \cref{lem:mms_agent_ub}, there is an instance where agent $n-1$ has a positive MMS and gets no more than $1/\lfloor (m-n+2)/2 \rfloor$ times her MMS.
Since $m-n$ is even, $1/\lfloor (m-n+2)/2 \rfloor = 2/(m-n+2)$.
This completes the proof.
\end{proof}

\Cref{thm:mms_regular,thm:mms_irregular} allow us to determine the best recursively balanced picking sequences with respect to MMS.

\begin{restatable}[Best Picking Sequences]{theorem}{mmsBest}
\label{thm:mms_best}
Let $n \geq 2$ and $m \geq n$.
Denote
\[
    \alpha_{\max} = \min\left\{\frac{\lfloor m/n \rfloor}{\lfloor m/n \rfloor n - n + 1}, \frac{\lceil m/n \rceil}{m-n+1}\right\}.
\]
Then, the following statements hold:
\begin{enumerate}[(a)]
    \item For every picking sequence $\pi \in \recbalseq$, there exists an instance~$\mathcal{I} \in \mathcal{I}_{n, m}$ such that some agent with a positive MMS gets no more than $\alpha_{\max}$ times her MMS in the allocation obtained with~$\pi$.
    \item There exists a picking sequence $\pi \in \recbalseq$ such that every agent is always guaranteed to get at least $\alpha_{\max}$ times her MMS in the allocation obtained with~$\pi$.
    \item For any picking sequence $\pi \in \recbalseq$, every agent is always guaranteed to get at least $\alpha_{\max}$ times her MMS in the allocation obtained with~$\pi$ if and only if $\pi$ satisfies the following:
    \begin{enumerate}[(i)]
        \item $\pi$ is regular, and
        \item for each $r \in \{2, \ldots, \lceil m/n \rceil\}$, agent $n$ gets to pick in round $r$ and the index of her $r$-th pick is at most $n + (r-1)/\alpha_{\max}$. 
    \end{enumerate}
\end{enumerate}
\end{restatable}

\begin{proof}
We first prove statement (a).

Consider first the instance~$\mathcal{I}_1 \in \mathcal{I}_{n, m}$ with the following utilities:
\begin{itemize}
    \item For $i \in N \setminus \{n\}$, $u_i(g_i) = m$ and $u_i(g_j) = 0$ if $j \neq i$.
    \item $u_n(g_j) = m$ if $j < n$.
    \item $u_n(g_j) = 1$ if $n \leq j \leq \lfloor m/n \rfloor n$.
    \item $u_n(g_j) = 0$ if $j > \lfloor m/n \rfloor n$.
\end{itemize}
Agents break ties in favor of lower-index goods.

We observe that agent~$n$'s MMS is exactly $u_n(\{g_n, \ldots, g_m\}) = \allowbreak \lfloor m/n \rfloor n - n + 1 \geq 1$, as obtained by the partition $\{\{g_1\}, \{g_2\}, \ldots, \{g_{n-1}\}, \{g_n, g_{n+1}, \ldots, g_{m}\}\}$. 
In each round $r \in \{1, \ldots, \lfloor m/n \rfloor\}$, agent $n$ gets a good with utility $1$. 
If there is an incomplete last round $r = \lfloor m/n \rfloor + 1$, either agent $n$ does not get to pick at all, or she picks a good with utility $0$. In both cases, agent $n$'s total utility is $\lfloor m/n \rfloor$. 
Hence, agent~$n$ gets exactly $\lfloor m/n \rfloor/(\lfloor m/n \rfloor n - n + 1)$ times her MMS.

Now, consider a different instance~$\mathcal{I}_2 \in \mathcal{I}_{n, m}$ with the following utilities:
\begin{itemize}
    \item For $i \in N \setminus \{n\}$, $u_i(g_i) = m$ and $u_i(g_j) = 0$ if $j \neq i$.
    \item $u_n(g_j) = m$ if $j < n$.
    \item $u_n(g_j) = 1$ if $n \leq j \leq m$.
\end{itemize}
Again, agents break ties in favor of lower-index goods.

In this instance, agent $n$'s MMS is exactly $u_n(\{g_n, \ldots, g_m\}) = m - n + 1 \geq 1$, as obtained by the~partition $\{\{g_1\}, \{g_2\}, \ldots, \{g_{n-1}\}, \{g_n, g_{n+1}, \ldots, g_{m}\}\}$. Moreover, in each round, agent $n$ gets a good with utility at most $1$, giving a total utility of at most $\lceil m/n \rceil$.
It follows that agent~$n$ gets at most $\lceil m/n \rceil/(m - n + 1)$ times her MMS.

Therefore, for one of the instances $\mathcal{I}_1$ and $\mathcal{I}_2$, agent $n$ gets no more than \[
    \alpha_{\max} = \min\left\{\frac{\lfloor m/n \rfloor}{\lfloor m/n \rfloor n - n + 1}, \frac{\lceil m/n \rceil}{m-n+1}\right\}
\] times her MMS.

Next, we prove statement (b) using statement (c), which will be proven later.
Consider the picking sequence $\pi$ where the sequence in the second round onwards is the reverse of the first round---that is, $\pi = (1, 2, \ldots, n \mid n, n-1, \ldots, 1 \mid \cdots \mid n, n-1, \ldots, 1 \mid n, n-1, \ldots, a_m)$ where $a_m = \lceil m/n \rceil n - m + 1$.
Observe that $\pi$ is regular. 
Indeed, otherwise the second round exists and has an even number of turns, so agent~$n-1$ gets to pick in the second round, which means that $\pi$ cannot be irregular.
We now prove that $\pi$ satisfies (ii).
Take any $r \in \{2, \ldots, \lceil m/n \rceil\}$, and let $t_r$ be the index of agent $n$'s $r$-th pick.
Note that $t_r = (r-1)n+1$ and $r \leq \lceil m/n \rceil \leq \lfloor m/n \rfloor+1$.
Then,
\begin{align*}
    \frac{r-1}{t_r-n} &= \frac{r-1}{(r-1)n+1-n} \\
    &= \frac{1}{n-\frac{n-1}{r-1}} \\
    &\geq \frac{1}{n-\frac{n-1}{\lfloor m/n \rfloor}} \tag{since $r-1 \leq \lfloor m/n \rfloor$} \\
    &= \frac{\lfloor m/n \rfloor}{\lfloor m/n \rfloor n-n+1} \\
    &\geq \alpha_{\max}.
\end{align*}
This means that $t_r \leq n + (r-1)/\alpha_{\max}$. 
Hence, $\pi$ also satisfies condition (ii).
By statement (c), under picking sequence $\pi$, every agent is guaranteed to get at least $\alpha_{\max}$ times her MMS.

We now prove the backward direction of statement (c).
Take any picking sequence $\pi \in \recbalseq$ that satisfies both conditions (i) and~(ii).
In particular, $\pi$ is regular.
We shall use \cref{thm:mms_regular}(a) to prove the statement.
Let $\pi_n = (t_1, \ldots, t_R)$ be agent $n$'s picking sequence under $\pi$ and $t_{R+1} = m+1$.
By condition (ii), we have $R = \lceil m/n \rceil$ and $t_r \leq n+(r-1)/\alpha_{\max}$ for each $r \in \{2, \ldots, R\}$.
It follows that
\begin{align*}
    &\min_{r \in \{2, \ldots, R+1\}} \frac{r-1}{t_r-n} \\
    &= \min \left\{\min_{r \in \{2, \ldots, R\}} \frac{r-1}{t_r-n}, \frac{R}{t_{R+1}-n} \right\} \\
    &\geq \min\left\{\min_{r \in \{2, \ldots, R\}} \frac{r-1}{(r-1)/\alpha_{\max}}, \frac{R}{t_{R+1}-n} \right\} \tag{by the assumption on $t_r$} \\
    &= \min\left\{\alpha_{\max}, \frac{\lceil m/n \rceil}{m-n+1} \right\} \tag{since $R = \lceil m/n \rceil$ and $t_{R+1} = m+1$} \\
    &= \alpha_{\max}. \tag{by definition of $\alpha_{\max}$}
\end{align*}
Hence, by \cref{thm:mms_regular}(a), every agent is guaranteed to get at least $\alpha_{\max}$ times her MMS.

It remains to prove the forward direction of statement (c).
Let $\pi \in \recbalseq$ be such that every agent is guaranteed to get at least $\alpha_{\max}$ times her MMS.
Let $\pi_n = (t_1, \ldots, t_R)$ be agent $n$'s picking sequence under $\pi$ and $t_{R+1} = m+1$.

Assume for the sake of contradiction that $\pi$ is irregular. 
Then, the second round exists and is incomplete. 
This means that $\lfloor m/n \rfloor = 1$ and $\lceil m/n \rceil = 2$. 
Hence, $\alpha_{\max} = \min\{1,2/(m-n+1)\} = 2/(m-n+1).$ 
On the other hand, by \cref{thm:mms_irregular}(b), there exists an instance where some agent with a positive MMS gets no more than $2/(m-n+2)$ times her MMS.
Since $2/(m-n+2) < 2/(m-n+1) = \alpha_{\max}$, this contradicts the assumption that every agent gets at least $\alpha_{\max}$ times her MMS.
It follows that $\pi$ must be regular.

We can now prove, by contradiction, that agent $n$ must pick in the last round $r = \lceil m/n \rceil$. Suppose otherwise. Since $\pi$ is recursively balanced, the last round where agent $n$ gets to pick is $R = \lfloor m/n \rfloor \leq \lceil m/n \rceil - 1$; that is, $m$ is not divisible by $n$.
Then, we have
\begin{align*}
    &\min_{r \in \{2, \ldots, R+1\}} \frac{r-1}{t_r-n} \\
    &\quad\leq \frac{R}{t_{R+1}-n} \\
    &\quad= \frac{\lfloor m/n \rfloor}{m-n+1} \tag{since $R = \lfloor m/n \rfloor$ and $t_{R+1} = m+1$} \\
    &\quad< \min\left\{\frac{\lfloor m/n \rfloor}{\lfloor m/n \rfloor n - n + 1}, \frac{\lceil m/n \rceil}{m-n+1}\right\} \tag{since $m > \lfloor m/n \rfloor n$ and $\lfloor m/n \rfloor < \lceil m/n \rceil$} \\
    &\quad= \alpha_{\max}.
\end{align*}
By \cref{lem:mms_agent_ub}, this implies that agent $n$ is not always guaranteed to get at least $\alpha_{\max}$ times her MMS, a contradiction.
Hence, agent~$n$ must pick in every round.

Lastly, we prove, again by contradiction, that $t_r \leq n+(r-1)/\alpha_{\max}$ for all $r \in \{2, \ldots, \lceil m/n \rceil\}$.
Suppose that there is some $r \in \{2, \ldots, \lceil m/n \rceil\}$ such that $t_r > n+(r-1)/\alpha_{\max}$.
This means that $\alpha_{\max} > (r-1)/(t_r-n)$.
Like in the previous paragraph, \cref{lem:mms_agent_ub} implies that agent $n$ is not guaranteed to get $\alpha_{\max}$ times her MMS, contradicting our earlier assumption.
Hence, it must hold that $t_r \leq n+(r-1)/\alpha_{\max}$ for all $r \in \{2, \ldots, \lceil m/n \rceil\}$. Together with the fact that agent $n$ picks in round $\lceil m/n \rceil$, this implies that condition (ii) is satisfied.
\end{proof}

At the other extreme, we characterize the recursively balanced picking sequences with the worst MMS guarantee.

\begin{restatable}[Worst Picking Sequences]{theorem}{mmsWorst}
\label{thm:mms_worst}
Let $n \geq 2$ and $m \geq n$.
Denote 
\[
    \alpha_{\min} = \max\left\{\frac{1}{n}, \frac{1}{m-n+1}\right\}.
\]
Then, the following statements hold:
\begin{enumerate}[(a)]
    \item For every picking sequence $\pi \in \recbalseq$, every agent is always guaranteed to get at least $\alpha_{\min}$ times her MMS in the allocation obtained with~$\pi$.
    \item There exists a picking sequence $\pi \in \recbalseq$ such that in some instance~$\mathcal{I} \in \mathcal{I}_{n, m}$, some agent with a positive MMS gets no more than $\alpha_{\min}$ times her MMS in the allocation obtained with~$\pi$.
    \item For any picking sequence $\pi \in \recbalseq$, some agent with a positive MMS gets no more than $\alpha_{\min}$ times her MMS in the allocation obtained with~$\pi$ in some instance~$\mathcal{I} \in \mathcal{I}_{n, m}$ if and only if $\pi$ satisfies the following:
    \begin{enumerate}[(i)]
        \item $m \leq 2n-1$ and agent $n$ only picks once, or
        \item $m \geq 2n$ and there is some round $r \in \{2, \ldots,  \lceil m/n \rceil\}$ consisting of at least $n-1$ turns where agent $n$ does not pick in the first $n-1$ turns.
    \end{enumerate}
\end{enumerate}
\end{restatable}

When $m \ge 2n-1$, we have $\alpha_{\min} = 1/n$.
In this case, \Cref{thm:mms_worst}(a) follows from Proposition~3.6 of \citet{AmanatidisBiMa18}, which shows that any EF1 allocation gives each agent at least $1/n$ times her MMS.

Note that round-robin always satisfies condition (i) or (ii) of \cref{thm:mms_worst}(c), and therefore has the worst MMS guarantee. 
More interestingly, balanced alternation also satisfies condition (ii) when $m \geq 3n-1$, since agent $n$ does not pick in the first $n-1$ turns of the third round.

\begin{proof}[Proof of \cref{thm:mms_worst}]
Note that $m-n+1 \leq n$ if and only if $m \leq 2n-1$. 
Therefore, we have
\begin{align}
\label{eq:alpha_min}
    \alpha_{\min} = \begin{cases}
        1/(m-n+1) & \text{if $m \leq 2n-1$}; \\
        1/n & \text{if $m \geq 2n$.}
    \end{cases}
\end{align}

We first prove statement (a).
Consider the case where $m \leq 2n-1$.
By \cref{lem:mms_agent_lb_small_m}, each agent $i \in N$ is guaranteed to get at least $1/\lfloor (m+1-i)/(n+1-i) \rfloor$ times her MMS.
We have
\begin{align*}
    \frac{1}{\left\lfloor \frac{m+1-i}{n+1-i} \right\rfloor} &\geq \frac{n+1-i}{m+1-i} \\
    &= \frac{1}{\frac{m-n}{n+1-i} + 1} \\
    &\geq \frac{1}{m-n + 1} \tag{since $i \leq n$} \\
    &= \alpha_{\min}. \tag{from \eqref{eq:alpha_min} and $m \leq 2n-1$}
\end{align*}
Hence, every agent $i \in N$ is guaranteed to get at least $\alpha_{\min}$ times her MMS.

Next, consider the case where $m \geq 2n$.
By \cref{lem:mms_agent_lb_large_m}, each agent $i \in N$ is guaranteed to get at least $(n+1-i)/(2n-i)$ times her MMS.
We have
\begin{align*}
    \frac{n+1-i}{2n-i} &= \frac{1}{\frac{n-1}{n+1-i}+1} \\
    &\geq \frac{1}{n} \tag{since $i \leq n$} \\
    &= \alpha_{\min}. \tag{from \eqref{eq:alpha_min} and $m \geq 2n$}
\end{align*}
Hence, every agent $i \in N$ is guaranteed to get at least $\alpha_{\min}$ times her MMS.

We prove statement (b) using statement (c), which will be proven later.
Consider the round-robin picking sequence $\roundrobin$.
If $m \leq 2n-1$, agent $n$ only picks once, while if $m \geq 2n$, agent $n$ picks in the $n$-th turn of the second round.
In both cases, $\roundrobin$ satisfies either condition~(i) or condition (ii) of statement (c).
Hence, it follows from statement (c) that under picking sequence $\roundrobin$, there exists an instance where some agent with a positive MMS gets no more than $\alpha_{\min}$ times her MMS.

We now prove the backward direction of statement (c).
Take any picking sequence $\pi \in \recbalseq$ satisfying either condition (i) or (ii).
Let agent $n$'s picking sequence be $\pi_n = (t_1, \ldots, t_R)$ and define $t_{R+1} = m+1$.

First, suppose that $m \leq 2n-1$. Then, condition (i) must be satisfied, and agent $n$ only picks once---that is, $\pi_n = (t_1) = (n)$.
In this case, we have
\begin{align*}
    \min_{r \in \{2, \ldots, R+1\}} \frac{r-1}{t_r-n} = \frac{2-1}{t_2-n} = \frac{1}{m-n+1} \leq \alpha_{\min}.
\end{align*}
Hence, by \cref{lem:mms_agent_ub}, there exists an instance where agent~$n$ has a positive MMS and gets no more than $\alpha_{\min}$ times her MMS.

Next, consider the case where $m \geq 2n$. 
Then, condition~(ii) must be satisfied, and there is some round $s \in \{2, \ldots, \lceil m/n \rceil\}$ consisting of at least $n-1$ turns where agent~$n$ does not pick in the first $n-1$ turns.
If agent~$n$ picks in round~$s$, then $t_s = sn$ and $s \in \{2, \ldots, R\}$.
Otherwise, agent~$n$ does not pick in round $s$. Then, $s = R+1$ and $m \leq sn-1$. Since round $s$ consists of at least $n-1$ turns, $m = sn-1 = (R+1)n-1$. Thus, $t_{R+1} = m+1 = (R+1)n$.
In either case, we have $t_s = sn$ and $s \in \{2, \ldots, R+1\}$.
It follows that
\begin{align*}
    \min_{r \in \{2, \ldots, R+1\}} \frac{r-1}{t_r-n} &\leq \frac{s-1}{t_s-n}
    = \frac{s-1}{sn-n} 
    = \frac{1}{n}
    \leq \alpha_{\min}.
\end{align*}
Again, by \cref{lem:mms_agent_ub}, there exists an instance where agent~$n$ has a positive MMS and gets no more than $\alpha_{\min}$ times her MMS.

It remains to prove the forward direction of statement (c).
We proceed by proving the contrapositive.
Take any picking sequence $\pi \in \recbalseq$ that satisfies neither condition (i) nor (ii).
Let agent $n$'s picking sequence be $\pi_n = (t_1, \ldots, t_R)$ and define $t_{R+1} = m+1$.

We first consider the case where $m \leq 2n-1$. Since condition (i) is not satisfied, agent $n$ must get to pick in the second round.
That is, $R = 2$ and $t_2 \leq m$.
Furthermore, $m \geq n+1$.
We have
\begin{align*}
    \min_{r \in \{2, \ldots, R+1\}} \frac{r-1}{t_r-n} &= \min \left\{\frac{1}{t_2-n}, \frac{2}{t_3-n}\right\} \\
    &\geq \min \left\{\frac{1}{m-n}, \frac{2}{m+1-n}\right\} \tag{since $t_2 \leq m$ and $t_3 = m+1$} \\
    &\geq \min \left\{\frac{1}{m-n}, \frac{2}{2(m-n)}\right\}  \tag{since $1 \leq m-n$} \\
    &= \frac{1}{m-n} \\
    &> \frac{1}{m-n+1} \\
    &= \alpha_{\min}. \tag{from \eqref{eq:alpha_min} and $m \leq 2n-1$}
\end{align*}
Hence, by \cref{lem:mms_agent_lb}, agent $n$ is guaranteed to get strictly more than $\alpha_{\min}$ times her MMS whenever her MMS is positive.

We now show that every agent $i \in N \setminus \{n\}$ is also guaranteed to get strictly more than $\alpha_{\min}$ times her MMS whenever her MMS is positive.
By \cref{lem:mms_agent_lb_small_m}, such an agent $i \in N \setminus \{n\}$ is guaranteed to get at least $1/\lfloor (m+1-i)/(n+1-i) \rfloor$ times her MMS. 
Note that
\begin{align*}
    \frac{1}{\left\lfloor \frac{m+1-i}{n+1-i} \right\rfloor} &\geq \frac{n+1-i}{m+1-i} \\
    &= \frac{1}{\frac{m-n}{n+1-i}+1} \\
    &> \frac{1}{m-n+1} \tag{since $i < n$} \\
    &= \alpha_{\min}. \tag{from \eqref{eq:alpha_min} and $m \leq 2n-1$}
\end{align*}
Hence, agent~$i$ gets strictly more than $\alpha_{\min}$ times her MMS whenever her MMS is positive.

Lastly, we consider the case where $m \geq 2n$.
Let $r \in \{2, \ldots, R\}$.
If $t_r \geq rn$, then agent $n$ picks in the $n$-th turn of round $r$, which implies that condition (ii) is satisfied, a contradiction.
Hence, $t_r \leq rn-1$.
Furthermore, if $t_{R+1} \geq (R+1)n$, then $m = (R+1)n-1$, so round $R+1$ contains $n-1$ turns but agent $n$ does not pick in round $R+1$. Again, this implies that condition (ii) is satisfied, a contradiction.
Hence, for every $r \in \{2, \ldots, R+1\}$, we have $t_r \leq rn-1$.
Furthermore, since $m \geq 2n$, it holds that $\pi$ is regular.
By \cref{thm:mms_regular}, every agent is always guaranteed to get at least $\alpha$ times her MMS, where
\begin{align*}
    \alpha &= \min_{r \in \{2, \ldots, R+1\}} \frac{r-1}{t_r-n} \\
    &\geq \min_{r \in \{2, \ldots, R+1\}} \frac{r-1}{rn-1-n} \tag{since $t_r \leq rn-1$} \\
    &> \min_{r \in \{2, \ldots, R+1\}} \frac{r-1}{(r-1)n} \\
    &= 1/n \\
    &= \alpha_{\min}. \tag{from \eqref{eq:alpha_min} and $m \geq 2n$}
\end{align*}
Hence, every agent with a positive MMS gets strictly more than $\alpha_{\min}$ times her MMS.
\end{proof}

\Cref{thm:mms_best,thm:mms_worst} yield a simple classification of the MMS guarantees of (recursively balanced) picking sequences for two agents.
Specifically, for each $m \geq 3$, there is a unique picking sequence with the best MMS guarantee, which lets agent~$2$ picks first in every round except the first round; all other picking sequences have the worst MMS guarantee.
This is elaborated in \Cref{cor:mms_twoagents}.

\begin{restatable}{corollary}{corMmsTwoagents}
\label{cor:mms_twoagents}
Let $n = 2$, $m \geq 3$, and $\pi \in \recbalseq$.
Then, the following statements hold:
\begin{enumerate}[(a)]
    \item If agent $2$ picks first in every round from the second round onwards in $\pi$, then each agent is always guaranteed to get at least $1/(2 - 1/\lfloor m/2 \rfloor)$ times her MMS.
    Moreover, this bound is tight.
    \item Otherwise, each agent is always guaranteed to get at least $1/2$ times her MMS.
    Moreover, this bound is tight.
\end{enumerate}
\end{restatable}

\begin{proof}
Suppose that agent $2$ picks first in the second round onwards of $\pi$.
For statement~(a), we begin by showing that
\begin{equation}
\label{eq:mms_twoagents}
    \frac{\lfloor m/2 \rfloor}{2\lfloor m/2 \rfloor - 1} \leq \frac{\lceil m/2 \rceil}{m-1}.
\end{equation}
When $m$ is even, both sides are equal. 
Assume now that $m$ is odd. 
In this case, $\lceil m/2 \rceil - \lfloor m/2 \rfloor = m - 2\lfloor m/2 \rfloor = 1$, so we have
\begin{align*}
    \frac{\lceil m/2 \rceil}{m-1} &= \frac{\lfloor m/2 \rfloor + 1}{2\lfloor m/2 \rfloor} \\
    &= \frac{1}{2 - \frac{2}{\lfloor m/2 \rfloor + 1}} \\
    &\geq \frac{1}{2 - \frac{2}{2\lfloor m/2 \rfloor}} \tag{since $\lfloor m/2 \rfloor \geq 1$} \\
    &= \frac{\lfloor m/2 \rfloor}{2\lfloor m/2 \rfloor - 1}.
\end{align*}
Hence, $\alpha_{\max}$ in \Cref{thm:mms_best} satisfies
\begin{align*}
    \alpha_{\max} &= \min\left\{\frac{\lfloor m/n \rfloor}{\lfloor m/n \rfloor n - n + 1}, \frac{\lceil m/n \rceil}{m-n+1}\right\} \\
    &= \frac{\lfloor m/2 \rfloor}{2\lfloor m/2 \rfloor-1} \tag{by \eqref{eq:mms_twoagents}} \\
    &= \frac{1}{2 - \frac{1}{\lfloor m/2 \rfloor}}.
\end{align*}
Now, note that $\pi$ is regular---the second round cannot be incomplete and have an even number of turns simultaneously.
Furthermore, for every round $r \in \{2, \ldots, \lceil m/n \rceil\}$, agent $2$ picks first, so her index is $t_r = 2(r-1)+1$. 
We have
\begin{align*}
    \frac{r-1}{t_r-2} &= \frac{r-1}{2(r-1)-1} \\
    &= \frac{1}{2-\frac{1}{r-1}} \\
    &\geq \frac{1}{2-\frac{1}{\lfloor m/2 \rfloor}} \tag{since $r \leq \lceil m/n \rceil \leq \lfloor m/2 \rfloor+1$} \\
    &= \alpha_{\max}.
\end{align*}
Hence, $t_r \leq 2 + (r-1)/\alpha_{\max}$.
This means that $\pi$ satisfies both conditions (i) and (ii) of \cref{thm:mms_best}(c).
Therefore, each agent is guaranteed to get $\alpha_{\max} = 1/(2-1/\lfloor m/2 \rfloor)$ times her MMS.
The tightness of this guarantee follows from \cref{thm:mms_best}(a).

Otherwise, suppose that agent $2$ does not pick first in some round after the first round.
The MMS guarantee in statement~(b) follows from \cref{thm:mms_worst}(a)---note that $m \geq 3$, and so $\alpha_{\min} = \max\{1/n, 1/(m-n+1)\} = 1/2$. 
For the tightness of this guarantee, observe that $\pi$ necessarily satisfies condition (i) or (ii) of \cref{thm:mms_worst}(c).
\end{proof}

When there are more than two agents, however, the picture is not as simple.
In particular, there is more than one picking sequence with the best MMS guarantee, and there are picking sequences with neither the best nor the worst MMS guarantee.
To illustrate these points, we consider the case with $3$ agents and $7$ goods in the following example.

\begin{example}
Let $n = 3$ and $m = 7$.

First, we find all picking sequences in $\mathcal{R}_{3, 7}$ with the best MMS guarantee.
By \cref{thm:mms_best}, $\alpha_{\max} = \min\{2/4, 3/5\} = 1/2$.
In order to guarantee that each agent always gets at least $\alpha_{\max} = 1/2$ times her MMS, a picking sequence must satisfy conditions (i) and (ii) of \cref{thm:mms_best}(c).
Condition~(i) is always satisfied since the second round is complete.
Condition (ii) is satisfied if and only if agent~$3$ gets to pick three times and the indices of agent~$3$'s second and third picks are at most $3 + 1/\alpha_{\max} = 5$ and $3 + 2/\alpha_{\max} = 7$ respectively.
We can therefore list all satisfying picking sequences as follows:
\[\begin{array}{ll}
    (1, 2, 3 \mid \mathbf{3}, 1, 2 \mid \mathbf{3}); & (1, 2, 3 \mid 1, \mathbf{3}, 2 \mid \mathbf{3}); \\
    (1, 2, 3 \mid \mathbf{3}, 2, 1 \mid \mathbf{3}); & (1, 2, 3 \mid 2, \mathbf{3}, 1 \mid \mathbf{3}).
\end{array}\]

Next, we find all picking sequences in $\mathcal{R}_{3, 7}$ with the worst MMS guarantee.
By \cref{thm:mms_worst}, $\alpha_{\min} = \max\{1/3, 1/5\} = 1/3$.
Since $m \geq 2n$, these picking sequences must satisfy condition (ii) of \cref{thm:mms_worst}(c), that is, agent $3$ must pick last in the second round.
All picking sequences satisfying this condition are as follows:
\[\begin{array}{lll}
    (1, 2, 3 \mid 1, 2, \mathbf{3} \mid 1); & (1, 2, 3 \mid 2, 1, \mathbf{3} \mid 1); & (1, 2, 3 \mid 1, 2, \mathbf{3} \mid 2); \\
    (1, 2, 3 \mid 2, 1, \mathbf{3} \mid 2); & (1, 2, 3 \mid 1, 2, \mathbf{3} \mid 3); & (1, 2, 3 \mid 2, 1, \mathbf{3} \mid 3).
\end{array}\]

In total, there are four picking sequences with the best MMS guarantee of $\alpha_{\max} = 1/2$, and six sequences with the worst MMS guarantee of $\alpha_{\min} = 1/3$. 
Since the number of picking sequences in $\mathcal{R}_{3, 7}$ is $3! \cdot 3 = 18$, there are eight picking sequences with neither the best nor the worst MMS guarantee, e.g., $(1, 2, 3 \mid 1, 3, 2 \mid 1)$.
\end{example}

\section{Conclusion}
\label{sec:conclusion}

In this paper, we have compared the fairness of recursively balanced picking sequences, all of which are known to guarantee envy-freeness up to one good (EF1).
We used two important measures, egalitarian welfare and approximate maximin share (MMS), and showed that they yield highly different results. 
On the one hand, all recursively balanced picking sequences fare equally well when evaluated using worst-case egalitarian welfare relative to other picking sequences. 
On the other hand, various recursively balanced picking sequences offer differing MMS guarantees, with the round-robin and balanced alternation sequences being among the worst. 
Interestingly, the sequences with the best MMS guarantee include those in which the agent who picks last in the first round always picks first in every subsequent round.
In light of their theoretical guarantees, we believe that these sequences merit consideration for practical adoption.

We conclude by proposing two directions for future research.
Firstly, it would be interesting to see whether our worst-case results continue to hold in the average case.
In this vein, \citet{BouveretLa11} computed the optimal picking sequences with respect to egalitarian (and utilitarian) welfare; however, their results are empirical and limited to small numbers of agents and goods.\footnote{Moreover, our measure is somewhat different, as we consider the egalitarian \emph{price} (see \Cref{sec:ew}).}
Note that average-case results rely on assumptions on the distributions of agents' utilities, which may vary across different applications.
Secondly, one could extend our analyses to picking sequences that are not recursively balanced.
While such picking sequences do not provide the EF1 guarantee, they may nevertheless perform well according to other measures.

\balance



\begin{acks}
This work was partially supported by the Singapore Ministry of Education under grant number MOE-T2EP20221-0001 and by an NUS Start-up Grant.
We thank the anonymous reviewers for their valuable comments.
\end{acks}



\bibliographystyle{ACM-Reference-Format} 
\bibliography{main-arxiv}


\begin{thebibliography}{32}


\ifx \showCODEN    \undefined \def \showCODEN     #1{\unskip}     \fi
\ifx \showDOI      \undefined \def \showDOI       #1{#1}\fi
\ifx \showISBNx    \undefined \def \showISBNx     #1{\unskip}     \fi
\ifx \showISBNxiii \undefined \def \showISBNxiii  #1{\unskip}     \fi
\ifx \showISSN     \undefined \def \showISSN      #1{\unskip}     \fi
\ifx \showLCCN     \undefined \def \showLCCN      #1{\unskip}     \fi
\ifx \shownote     \undefined \def \shownote      #1{#1}          \fi
\ifx \showarticletitle \undefined \def \showarticletitle #1{#1}   \fi
\ifx \showURL      \undefined \def \showURL       {\relax}        \fi
\providecommand\bibfield[2]{#2}
\providecommand\bibinfo[2]{#2}
\providecommand\natexlab[1]{#1}
\providecommand\showeprint[2][]{arXiv:#2}

\bibitem[\protect\citeauthoryear{Akrami and Garg}{Akrami and Garg}{2024}]%
        {AkramiGa24}
\bibfield{author}{\bibinfo{person}{Hannaneh Akrami} {and} \bibinfo{person}{Jugal Garg}.} \bibinfo{year}{2024}\natexlab{}.
\newblock \showarticletitle{Breaking the $3/4$ barrier for approximate maximin share}. In \bibinfo{booktitle}{\emph{Proceedings of the 35th ACM-SIAM Symposium on Discrete Algorithms (SODA)}}. \bibinfo{pages}{74--91}.
\newblock


\bibitem[\protect\citeauthoryear{Amanatidis, Aziz, Birmpas, Filos-Ratsikas, Li, Moulin, Voudouris, and Wu}{Amanatidis et~al\mbox{.}}{2023}]%
        {AmanatidisAzBi23}
\bibfield{author}{\bibinfo{person}{Georgios Amanatidis}, \bibinfo{person}{Haris Aziz}, \bibinfo{person}{Georgios Birmpas}, \bibinfo{person}{Aris Filos-Ratsikas}, \bibinfo{person}{Bo Li}, \bibinfo{person}{Hervé Moulin}, \bibinfo{person}{Alexandros~A. Voudouris}, {and} \bibinfo{person}{Xiaowei Wu}.} \bibinfo{year}{2023}\natexlab{}.
\newblock \showarticletitle{Fair division of indivisible goods: Recent progress and open questions}.
\newblock \bibinfo{journal}{\emph{Artificial Intelligence}}  \bibinfo{volume}{322} (\bibinfo{year}{2023}), \bibinfo{pages}{103965}.
\newblock


\bibitem[\protect\citeauthoryear{Amanatidis, Birmpas, Fusco, Lazos, Leonardi, and Reiffenh\"{a}user}{Amanatidis et~al\mbox{.}}{2024}]%
        {AmanatidisBiFu24}
\bibfield{author}{\bibinfo{person}{Georgios Amanatidis}, \bibinfo{person}{Georgios Birmpas}, \bibinfo{person}{Federico Fusco}, \bibinfo{person}{Philip Lazos}, \bibinfo{person}{Stefano Leonardi}, {and} \bibinfo{person}{Rebecca Reiffenh\"{a}user}.} \bibinfo{year}{2024}\natexlab{}.
\newblock \showarticletitle{Allocating indivisible goods to strategic agents: Pure {Nash} equilibria and fairness}.
\newblock \bibinfo{journal}{\emph{Mathematics of Operations Research}} \bibinfo{volume}{49}, \bibinfo{number}{4} (\bibinfo{year}{2024}), \bibinfo{pages}{2425--2445}.
\newblock


\bibitem[\protect\citeauthoryear{Amanatidis, Birmpas, and Markakis}{Amanatidis et~al\mbox{.}}{2018}]%
        {AmanatidisBiMa18}
\bibfield{author}{\bibinfo{person}{Georgios Amanatidis}, \bibinfo{person}{Georgios Birmpas}, {and} \bibinfo{person}{Evangelos Markakis}.} \bibinfo{year}{2018}\natexlab{}.
\newblock \showarticletitle{Comparing approximate relaxations of envy-freeness}. In \bibinfo{booktitle}{\emph{Proceedings of the 27th International Joint Conference on Artificial Intelligence (IJCAI)}}. \bibinfo{pages}{42--48}.
\newblock


\bibitem[\protect\citeauthoryear{Aumann and Dombb}{Aumann and Dombb}{2015}]%
        {AumannDo15}
\bibfield{author}{\bibinfo{person}{Yonatan Aumann} {and} \bibinfo{person}{Yair Dombb}.} \bibinfo{year}{2015}\natexlab{}.
\newblock \showarticletitle{The efficiency of fair division with connected pieces}.
\newblock \bibinfo{journal}{\emph{ACM Transactions on Economics and Computation}} \bibinfo{volume}{3}, \bibinfo{number}{4} (\bibinfo{year}{2015}), \bibinfo{pages}{23:1--23:16}.
\newblock


\bibitem[\protect\citeauthoryear{Aziz, Bouveret, Lang, and Mackenzie}{Aziz et~al\mbox{.}}{2017}]%
        {AzizBoLa17}
\bibfield{author}{\bibinfo{person}{Haris Aziz}, \bibinfo{person}{Sylvain Bouveret}, \bibinfo{person}{Jérôme Lang}, {and} \bibinfo{person}{Simon Mackenzie}.} \bibinfo{year}{2017}\natexlab{}.
\newblock \showarticletitle{Complexity of manipulating sequential allocation}. In \bibinfo{booktitle}{\emph{Proceedings of the 31st AAAI Conference on Artificial Intelligence (AAAI)}}. \bibinfo{pages}{328--334}.
\newblock


\bibitem[\protect\citeauthoryear{Aziz, Freeman, Shah, and Vaish}{Aziz et~al\mbox{.}}{2024}]%
        {AzizFrSh24}
\bibfield{author}{\bibinfo{person}{Haris Aziz}, \bibinfo{person}{Rupert Freeman}, \bibinfo{person}{Nisarg Shah}, {and} \bibinfo{person}{Rohit Vaish}.} \bibinfo{year}{2024}\natexlab{}.
\newblock \showarticletitle{Best of both worlds: Ex ante and ex post fairness in resource allocation}.
\newblock \bibinfo{journal}{\emph{Operations Research}} \bibinfo{volume}{72}, \bibinfo{number}{4} (\bibinfo{year}{2024}), \bibinfo{pages}{1674--1688}.
\newblock


\bibitem[\protect\citeauthoryear{Aziz, Walsh, and Xia}{Aziz et~al\mbox{.}}{2015}]%
        {AzizWaXi15}
\bibfield{author}{\bibinfo{person}{Haris Aziz}, \bibinfo{person}{Toby Walsh}, {and} \bibinfo{person}{Lirong Xia}.} \bibinfo{year}{2015}\natexlab{}.
\newblock \showarticletitle{Possible and necessary allocations via sequential mechanisms}. In \bibinfo{booktitle}{\emph{Proceedings of the 24th International Joint Conference on Artificial Intelligence (IJCAI)}}. \bibinfo{pages}{468--474}.
\newblock


\bibitem[\protect\citeauthoryear{Baumeister, Bouveret, Lang, Nguyen, Nguyen, Rothe, and Saffidine}{Baumeister et~al\mbox{.}}{2017}]%
        {BaumeisterBoLa17}
\bibfield{author}{\bibinfo{person}{Dorothea Baumeister}, \bibinfo{person}{Sylvain Bouveret}, \bibinfo{person}{Jérôme Lang}, \bibinfo{person}{Nhan-Tam Nguyen}, \bibinfo{person}{Trung~Thanh Nguyen}, \bibinfo{person}{Jörg Rothe}, {and} \bibinfo{person}{Abdallah Saffidine}.} \bibinfo{year}{2017}\natexlab{}.
\newblock \showarticletitle{Positional scoring-based allocation of indivisible goods}.
\newblock \bibinfo{journal}{\emph{Autonomous Agents and Multi-Agent Systems}} \bibinfo{volume}{31}, \bibinfo{number}{3} (\bibinfo{year}{2017}), \bibinfo{pages}{628--655}.
\newblock


\bibitem[\protect\citeauthoryear{Bei, Lu, Manurangsi, and Suksompong}{Bei et~al\mbox{.}}{2021}]%
        {BeiLuMa21}
\bibfield{author}{\bibinfo{person}{Xiaohui Bei}, \bibinfo{person}{Xinhang Lu}, \bibinfo{person}{Pasin Manurangsi}, {and} \bibinfo{person}{Warut Suksompong}.} \bibinfo{year}{2021}\natexlab{}.
\newblock \showarticletitle{The price of fairness for indivisible goods}.
\newblock \bibinfo{journal}{\emph{Theory of Computing Systems}} \bibinfo{volume}{65}, \bibinfo{number}{7} (\bibinfo{year}{2021}), \bibinfo{pages}{1069--1093}.
\newblock


\bibitem[\protect\citeauthoryear{Bouveret, Gilbert, Lang, and Méroué}{Bouveret et~al\mbox{.}}{2025}]%
        {BouveretGiLa25}
\bibfield{author}{\bibinfo{person}{Sylvain Bouveret}, \bibinfo{person}{Hugo Gilbert}, \bibinfo{person}{Jérôme Lang}, {and} \bibinfo{person}{Guillaume Méroué}.} \bibinfo{year}{2025}\natexlab{}.
\newblock \showarticletitle{Constrained serial dictatorships can be fair}. In \bibinfo{booktitle}{\emph{Proceedings of the 34th International Joint Conference on Artificial Intelligence (IJCAI)}}. \bibinfo{pages}{3762--3770}.
\newblock


\bibitem[\protect\citeauthoryear{Bouveret and Lang}{Bouveret and Lang}{2011}]%
        {BouveretLa11}
\bibfield{author}{\bibinfo{person}{Sylvain Bouveret} {and} \bibinfo{person}{J\'{e}r\^{o}me Lang}.} \bibinfo{year}{2011}\natexlab{}.
\newblock \showarticletitle{A general elicitation-free protocol for allocating indivisible goods}. In \bibinfo{booktitle}{\emph{Proceedings of the 22nd International Joint Conference on Artificial Intelligence (IJCAI)}}. \bibinfo{pages}{73--78}.
\newblock


\bibitem[\protect\citeauthoryear{Bouveret and Lang}{Bouveret and Lang}{2014}]%
        {BouveretLa14}
\bibfield{author}{\bibinfo{person}{Sylvain Bouveret} {and} \bibinfo{person}{J\'{e}r\^{o}me Lang}.} \bibinfo{year}{2014}\natexlab{}.
\newblock \showarticletitle{Manipulating picking sequences}. In \bibinfo{booktitle}{\emph{Proceedings of the 21st European Conference on Artificial Intelligence (ECAI)}}. \bibinfo{pages}{141--146}.
\newblock


\bibitem[\protect\citeauthoryear{Brams, Edelman, and Fishburn}{Brams et~al\mbox{.}}{2003}]%
        {BramsEdFi03}
\bibfield{author}{\bibinfo{person}{Steven~J. Brams}, \bibinfo{person}{Paul~H. Edelman}, {and} \bibinfo{person}{Peter~C. Fishburn}.} \bibinfo{year}{2003}\natexlab{}.
\newblock \showarticletitle{Fair division of indivisible items}.
\newblock \bibinfo{journal}{\emph{Theory and Decision}} \bibinfo{volume}{55}, \bibinfo{number}{2} (\bibinfo{year}{2003}), \bibinfo{pages}{147--180}.
\newblock


\bibitem[\protect\citeauthoryear{Brams and Ismail}{Brams and Ismail}{2021}]%
        {BramsIs21}
\bibfield{author}{\bibinfo{person}{Steven~J. Brams} {and} \bibinfo{person}{Mehmet~S. Ismail}.} \bibinfo{year}{2021}\natexlab{}.
\newblock \showarticletitle{Fairer chess: a reversal of two opening moves in chess creates balance between white and black}. In \bibinfo{booktitle}{\emph{Proceedings of the 3rd IEEE Conference on Games (CoG)}}. \bibinfo{pages}{1--4}.
\newblock


\bibitem[\protect\citeauthoryear{Brams and Taylor}{Brams and Taylor}{2000}]%
        {BramsTa00}
\bibfield{author}{\bibinfo{person}{Steven~J. Brams} {and} \bibinfo{person}{Alan~D. Taylor}.} \bibinfo{year}{2000}\natexlab{}.
\newblock \bibinfo{booktitle}{\emph{The Win-Win Solution: Guaranteeing Fair Shares to Everybody}}.
\newblock \bibinfo{publisher}{W. W. Norton {\&} Company}.
\newblock


\bibitem[\protect\citeauthoryear{Budish and Cantillon}{Budish and Cantillon}{2012}]%
        {BudishCa12}
\bibfield{author}{\bibinfo{person}{Eric Budish} {and} \bibinfo{person}{Estelle Cantillon}.} \bibinfo{year}{2012}\natexlab{}.
\newblock \showarticletitle{The multi-unit assignment problem: Theory and evidence from course allocation at {Harvard}}.
\newblock \bibinfo{journal}{\emph{The American Economic Review}} \bibinfo{volume}{102}, \bibinfo{number}{5} (\bibinfo{year}{2012}), \bibinfo{pages}{2237--2271}.
\newblock


\bibitem[\protect\citeauthoryear{Caragiannis, Kaklamanis, Kanellopoulos, and Kyropoulou}{Caragiannis et~al\mbox{.}}{2012}]%
        {CaragiannisKaKa12}
\bibfield{author}{\bibinfo{person}{Ioannis Caragiannis}, \bibinfo{person}{Christos Kaklamanis}, \bibinfo{person}{Panagiotis Kanellopoulos}, {and} \bibinfo{person}{Maria Kyropoulou}.} \bibinfo{year}{2012}\natexlab{}.
\newblock \showarticletitle{The efficiency of fair division}.
\newblock \bibinfo{journal}{\emph{Theory of Computing Systems}} \bibinfo{volume}{50}, \bibinfo{number}{4} (\bibinfo{year}{2012}), \bibinfo{pages}{589--610}.
\newblock


\bibitem[\protect\citeauthoryear{Celine, Dzulfikar, and Koswara}{Celine et~al\mbox{.}}{2023}]%
        {CelineDzKo23}
\bibfield{author}{\bibinfo{person}{Karen~Frilya Celine}, \bibinfo{person}{Muhammad~Ayaz Dzulfikar}, {and} \bibinfo{person}{Ivan~Adrian Koswara}.} \bibinfo{year}{2023}\natexlab{}.
\newblock \showarticletitle{Egalitarian price of fairness for indivisible goods}. In \bibinfo{booktitle}{\emph{Proceedings of the 20th Pacific Rim International Conference on Artificial Intelligence (PRICAI)}}. \bibinfo{pages}{23--28}.
\newblock


\bibitem[\protect\citeauthoryear{Chakraborty, Schmidt-Kraepelin, and Suksompong}{Chakraborty et~al\mbox{.}}{2021}]%
        {ChakrabortyScSu21}
\bibfield{author}{\bibinfo{person}{Mithun Chakraborty}, \bibinfo{person}{Ulrike Schmidt-Kraepelin}, {and} \bibinfo{person}{Warut Suksompong}.} \bibinfo{year}{2021}\natexlab{}.
\newblock \showarticletitle{Picking sequences and monotonicity in weighted fair division}.
\newblock \bibinfo{journal}{\emph{Artificial Intelligence}}  \bibinfo{volume}{301} (\bibinfo{year}{2021}), \bibinfo{pages}{103578}.
\newblock


\bibitem[\protect\citeauthoryear{Demko and Hill}{Demko and Hill}{1988}]%
        {DemkoHi88}
\bibfield{author}{\bibinfo{person}{Stephen Demko} {and} \bibinfo{person}{Theodore~P. Hill}.} \bibinfo{year}{1988}\natexlab{}.
\newblock \showarticletitle{Equitable distribution of indivisible objects}.
\newblock \bibinfo{journal}{\emph{Mathematical Social Sciences}} \bibinfo{volume}{16}, \bibinfo{number}{2} (\bibinfo{year}{1988}), \bibinfo{pages}{145--158}.
\newblock


\bibitem[\protect\citeauthoryear{Ghodsi, Hajiaghayi, Seddighin, Seddighin, and Yami}{Ghodsi et~al\mbox{.}}{2021}]%
        {GhodsiHaSe21}
\bibfield{author}{\bibinfo{person}{Mohammad Ghodsi}, \bibinfo{person}{Mohammad~Taghi Hajiaghayi}, \bibinfo{person}{Masoud Seddighin}, \bibinfo{person}{Saeed Seddighin}, {and} \bibinfo{person}{Hadi Yami}.} \bibinfo{year}{2021}\natexlab{}.
\newblock \showarticletitle{Fair allocation of indivisible goods: Improvement}.
\newblock \bibinfo{journal}{\emph{Mathematics of Operations Research}} \bibinfo{volume}{46}, \bibinfo{number}{3} (\bibinfo{year}{2021}), \bibinfo{pages}{1038--1053}.
\newblock


\bibitem[\protect\citeauthoryear{Gourv{\`e}s, Lesca, and Wilczynski}{Gourv{\`e}s et~al\mbox{.}}{2021}]%
        {GourvesLeWi21}
\bibfield{author}{\bibinfo{person}{Laurent Gourv{\`e}s}, \bibinfo{person}{Julien Lesca}, {and} \bibinfo{person}{Ana{\"e}lle Wilczynski}.} \bibinfo{year}{2021}\natexlab{}.
\newblock \showarticletitle{On fairness via picking sequences in allocation of indivisible goods}. In \bibinfo{booktitle}{\emph{Proceedings of the 7th International Conference on Algorithmic Decision Theory (ADT)}}. \bibinfo{pages}{258--272}.
\newblock


\bibitem[\protect\citeauthoryear{Kalinowski, Narodytska, and Walsh}{Kalinowski et~al\mbox{.}}{2013a}]%
        {KalinowskiNaWa13a}
\bibfield{author}{\bibinfo{person}{Thomas Kalinowski}, \bibinfo{person}{Nina Narodytska}, {and} \bibinfo{person}{Toby Walsh}.} \bibinfo{year}{2013}\natexlab{a}.
\newblock \showarticletitle{A social welfare optimal sequential allocation procedure}. In \bibinfo{booktitle}{\emph{Proceedings of the 23rd International Joint Conference on Artificial Intelligence (IJCAI)}}. \bibinfo{pages}{227--233}.
\newblock


\bibitem[\protect\citeauthoryear{Kalinowski, Narodytska, Walsh, and Xia}{Kalinowski et~al\mbox{.}}{2013b}]%
        {KalinowskiNaWa13b}
\bibfield{author}{\bibinfo{person}{Thomas Kalinowski}, \bibinfo{person}{Nina Narodytska}, \bibinfo{person}{Toby Walsh}, {and} \bibinfo{person}{Lirong Xia}.} \bibinfo{year}{2013}\natexlab{b}.
\newblock \showarticletitle{Strategic behavior when allocating indivisible goods sequentially}. In \bibinfo{booktitle}{\emph{Proceedings of the 17th AAAI Conference on Artificial Intelligence (AAAI)}}. \bibinfo{pages}{452--458}.
\newblock


\bibitem[\protect\citeauthoryear{Kohler and Chandrasekaran}{Kohler and Chandrasekaran}{1971}]%
        {KohlerCh71}
\bibfield{author}{\bibinfo{person}{David~A. Kohler} {and} \bibinfo{person}{R. Chandrasekaran}.} \bibinfo{year}{1971}\natexlab{}.
\newblock \showarticletitle{A class of sequential games}.
\newblock \bibinfo{journal}{\emph{Operations Research}} \bibinfo{volume}{19}, \bibinfo{number}{2} (\bibinfo{year}{1971}), \bibinfo{pages}{270--277}.
\newblock


\bibitem[\protect\citeauthoryear{Kurokawa, Procaccia, and Wang}{Kurokawa et~al\mbox{.}}{2018}]%
        {KurokawaPrWa18}
\bibfield{author}{\bibinfo{person}{David Kurokawa}, \bibinfo{person}{Ariel~D. Procaccia}, {and} \bibinfo{person}{Junxing Wang}.} \bibinfo{year}{2018}\natexlab{}.
\newblock \showarticletitle{Fair enough: Guaranteeing approximate maximin shares}.
\newblock \bibinfo{journal}{\emph{J. ACM}} \bibinfo{volume}{65}, \bibinfo{number}{2} (\bibinfo{year}{2018}), \bibinfo{pages}{8:1--8:27}.
\newblock


\bibitem[\protect\citeauthoryear{Li, Liu, Lu, Tao, and Tao}{Li et~al\mbox{.}}{2024}]%
        {LiLiLu24}
\bibfield{author}{\bibinfo{person}{Zihao Li}, \bibinfo{person}{Shengxin Liu}, \bibinfo{person}{Xinhang Lu}, \bibinfo{person}{Biaoshuai Tao}, {and} \bibinfo{person}{Yichen Tao}.} \bibinfo{year}{2024}\natexlab{}.
\newblock \showarticletitle{A complete landscape for the price of envy-freeness}. In \bibinfo{booktitle}{\emph{Proceedings of the 23rd International Conference on Autonomous Agents and Multiagent Systems (AAMAS)}}. \bibinfo{pages}{1183--1191}.
\newblock


\bibitem[\protect\citeauthoryear{Li, Manurangsi, Scarlett, and Suksompong}{Li et~al\mbox{.}}{2025}]%
        {LiMaSc25}
\bibfield{author}{\bibinfo{person}{Zihan Li}, \bibinfo{person}{Pasin Manurangsi}, \bibinfo{person}{Jonathan Scarlett}, {and} \bibinfo{person}{Warut Suksompong}.} \bibinfo{year}{2025}\natexlab{}.
\newblock \showarticletitle{Complexity of round-robin allocation with potentially noisy queries}.
\newblock \bibinfo{journal}{\emph{Information and Computation}}  \bibinfo{volume}{306} (\bibinfo{year}{2025}), \bibinfo{pages}{105532}.
\newblock


\bibitem[\protect\citeauthoryear{Moulin}{Moulin}{2003}]%
        {Moulin03}
\bibfield{author}{\bibinfo{person}{Herv\'{e} Moulin}.} \bibinfo{year}{2003}\natexlab{}.
\newblock \bibinfo{booktitle}{\emph{Fair Division and Collective Welfare}}.
\newblock \bibinfo{publisher}{MIT Press}.
\newblock


\bibitem[\protect\citeauthoryear{Suksompong}{Suksompong}{2019}]%
        {Suksompong19}
\bibfield{author}{\bibinfo{person}{Warut Suksompong}.} \bibinfo{year}{2019}\natexlab{}.
\newblock \showarticletitle{Fairly allocating contiguous blocks of indivisible items}.
\newblock \bibinfo{journal}{\emph{Discrete Applied Mathematics}}  \bibinfo{volume}{260} (\bibinfo{year}{2019}), \bibinfo{pages}{227--236}.
\newblock


\bibitem[\protect\citeauthoryear{Tominaga, Todo, and Yokoo}{Tominaga et~al\mbox{.}}{2016}]%
        {TominagaToYo16}
\bibfield{author}{\bibinfo{person}{Yuto Tominaga}, \bibinfo{person}{Taiki Todo}, {and} \bibinfo{person}{Makoto Yokoo}.} \bibinfo{year}{2016}\natexlab{}.
\newblock \showarticletitle{Manipulations in two-agent sequential allocation with random sequences}. In \bibinfo{booktitle}{\emph{Proceedings of the 15th International Conference on Autonomous Agents and Multiagent Systems (AAMAS)}}. \bibinfo{pages}{141--149}.
\newblock


\end{thebibliography}


\clearpage
\appendix

\end{document}